\documentclass[paper=letter,abstracton]{scrartcl}
\def\jnPDFEXPORTVERSION{arxiv}

\PassOptionsToPackage{hyphens}{url}\usepackage{hyperref}

\usepackage{ifthen}

\ifthenelse{\equal{\jnPDFEXPORTVERSION}{springer}}{

}{

    \usepackage{lmodern}
    \usepackage{graphicx}

    \usepackage{amsthm}
    \newtheorem{theorem}{Theorem}
    \newtheorem{corollary}{Corollary}
    \newtheorem{definition}{Definition}

}

\usepackage{jn_preamble_pdflatex}
\usepackage{jn_notation}
\usepackage{jn_colors_parula}
\usepackage{jn_tikz}
\usepackage{jn_pgfplots}

\usepackage{shortsym}

\usetikzlibrary{arrows.meta}
\usetikzlibrary{shapes}
\usetikzlibrary{calc}
\usetikzlibrary{positioning}
\usepgfplotslibrary{fillbetween}

\allowdisplaybreaks

\usepackage{fancyvrb}

\usepackage[inline]{enumitem}

\usepackage{algorithmicx}
\usepackage{algpseudocode}
\usepackage{algorithm}
\newref{algo}{
    name={Alg.~},
    Name={Alg.~},
    names={Algs~},
    Names={Algs~},
}

\newcommand*\circled[1]{\tikz[baseline=(char.base)]{
            \node[shape=circle,draw,inner sep=2pt] (char) {\footnotesize #1};}}

\tikzset{
    txPromised/.append style={
        dashed,
    },
    txDelivered/.append style={
        solid,
        thick,
    },
    txNone/.append style={
        draw=none,
    },
}

\tikzstyle{fcFunds} = [circle,draw,node distance=0.75cm and 1cm]
\tikzstyle{fcDecision} = [rectangle,draw,node distance=0.75cm and 1cm,text width=3.5cm,text centered]
\tikzstyle{fcPayout} = [node distance=0.75cm and 1cm]
\tikzstyle{fcLine} = [-Latex]

\pgfplotsset{
    discard if not/.style 2 args={
        x filter/.code={
            \edef\tempa{\thisrow{#1}}
            \edef\tempb{#2}
            \ifx\tempa\tempb
            \else
                
            \fi
        }
    },
    myresultsplot01/.style={
        legend style={
            at={(0.5,1.15)},
            anchor=north,
            legend columns=-1,
            draw=none,
            /tikz/every even column/.append style={column sep=1em},
            cells={align=left},
            font=\footnotesize,   %
        },
        enlarge x limits=0.04,
        ybar,
        bar width=1.1mm,
        xtick=data,
        width=\linewidth,
        height=0.6\linewidth,
        ymin=0,
        grid=both,
        minor y tick num=4,
        minor x tick num=1,
        major grid style={solid,draw=gray!50},
        minor grid style={densely dotted,draw=gray!50},
        label style={font=\footnotesize},   %
        tick label style={font=\footnotesize},   %
    },
}

\DeclareSIUnit{\tx}{tx}
\DeclareSIUnit{\nde}{node}
\DeclareSIUnit{\xrp}{funds}
\newcommand{\drawrandom}[0]{\ensuremath{\xleftarrow{\mathrm{R}}}}
\newcommand{\cryptoSK}[1]{\ensuremath{\mathsf{sk}_{#1}}}
\newcommand{\cryptoPK}[1]{\ensuremath{\mathsf{pk}_{#1}}}

\newcommand{\cryptoSign}[2]{\ensuremath{\mathsf{sig}_{#1}\left(#2\right)}}

\mathlig{!<}{\stackrel{!}{<}}

\DeclareUnicodeCharacter{1F4C6}{\includegraphics[height=0.8em]{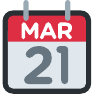}}
\newcommand{\emojiabstl}[0]{📆}
\DeclareUnicodeCharacter{23F3}{\includegraphics[height=0.8em]{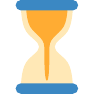}}
\newcommand{\emojireltl}[0]{⏳}
\DeclareUnicodeCharacter{1F512}{\includegraphics[height=0.8em]{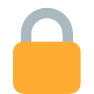}}
\newcommand{\emojisigcheck}[0]{🔒}
\DeclareUnicodeCharacter{1F58B}{\includegraphics[height=0.8em]{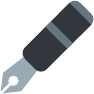}}
\newcommand{\emojisig}[0]{🖋}
\DeclareUnicodeCharacter{1F9E9}{\includegraphics[height=0.8em]{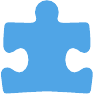}}
\newcommand{\emojihash}[0]{🧩}
\DeclareUnicodeCharacter{1F197}{\includegraphics[height=0.8em]{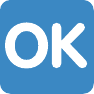}}
\newcommand{\emojipreimage}[0]{🆗}

\DeclareUnicodeCharacter{1F468}{\includegraphics[height=0.8em]{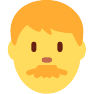}}
\newcommand{\emojibob}[0]{👨}
\DeclareUnicodeCharacter{1F469}{\includegraphics[height=0.8em]{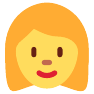}}
\newcommand{\emojialice}[0]{👩}
\DeclareUnicodeCharacter{1F464}{\includegraphics[height=0.8em]{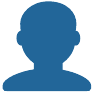}}
\newcommand{\emojiingrid}[0]{👤}
\DeclareUnicodeCharacter{1F465}{\includegraphics[height=0.8em]{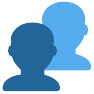}}

\usepackage{xcolor}

\ifthenelse{\equal{\jnPDFEXPORTVERSION}{springer}}{
    \raggedbottom
}{
    \raggedbottom
}

\begin{document}

\ifthenelse{\equal{\jnPDFEXPORTVERSION}{springer}}{
    \title{%
    Boomerang: Redundancy Improves Latency and Throughput in Payment-Channel Networks%
    }
}{
    \title{%
    Boomerang: Redundancy\\Improves Latency and Throughput\\in Payment-Channel Networks%
    }
}

\ifthenelse{\equal{\jnPDFEXPORTVERSION}{springer}}{
    \titlerunning{%
    Boomerang: Redundancy in Payment-Channel Networks%
    }
}{}

\author{%
Vivek Bagaria\textsuperscript{*}%
\and
Joachim Neu\textsuperscript{*}%
\and
David Tse%
}

\ifthenelse{\equal{\jnPDFEXPORTVERSION}{springer}}{
    \authorrunning{V. Bagaria \and J. Neu \and D. Tse}
}{}

\ifthenelse{\equal{\jnPDFEXPORTVERSION}{springer}}{
    \institute{%
    Stanford University, Stanford, USA%
    \\%
    \email{\{vbagaria,jneu,dntse\}@stanford.edu}%
    }
}{
    \publishers{\normalsize%
    Stanford University, Stanford, USA%
    \\%
    \texttt{\{vbagaria,jneu,dntse\}@stanford.edu}%
    }
    \date{}
}

\maketitle

\begingroup
\renewcommand\thefootnote{*}%
\footnotetext{Contributed equally and listed alphabetically.}
\endgroup

\begin{abstract}
In multi-path routing schemes
for payment-channel networks,
Alice transfers funds to Bob
by splitting them into partial payments
and routing them along multiple paths.
Undisclosed channel balances
and mismatched transaction fees
cause delays and failures
on
some
payment paths.
For atomic transfer schemes,
these straggling paths stall the whole transfer.
We show that the latency of transfers reduces when redundant payment paths are added.
This
frees up
liquidity
in
payment channels
and hence increases the throughput of the network.
We devise \emph{Boomerang},
a
generic
technique to be used on top of multi-path routing schemes to construct
redundant payment paths free of counterparty risk.
In our experiments,
applying Boomerang to
a baseline routing scheme
leads to
\SI{40}{\percent} latency reduction and
2x throughput increase.
We build on ideas from publicly verifiable secret sharing,
such that
Alice learns a secret of Bob
iff Bob overdraws funds from the redundant paths.
Funds are forwarded using Boomerang contracts,
which allow Alice to revert the transfer
iff she has learned Bob's secret.
We implement the Boomerang contract in Bitcoin Script.

\ifthenelse{\equal{\jnPDFEXPORTVERSION}{springer}}{
}{}
\end{abstract}

\section{Introduction}

\subsection{Payment Channels and Networks}
\label{sec:pcs-and-pns}

Blockchains
provide a method for maintaining a distributed ledger
in a decentralized and trustless fashion \cite{nakamoto_bitcoin:_2008}.
However, these so called layer-1 (L1) consensus mechanisms suffer
from low throughput
and high confirmation latency.
From a scaling perspective it would be desirable
if transactions could be processed `locally'
by only the involved participants.

Payment channels (PCs)
\cite{%
decker_fast_2015,%
poon_bitcoin_2016-1,%
decker_eltoo:_2018%
}
provide this. %
Once two participants have established a PC on-chain, %
they can %
transact through the channel off-chain without involving L1 every time.
Briefly,
the mechanism underlying the different PC implementations is as follows.
The participants escrow a pool of funds on-chain
which they can spend only jointly.
They can perform a transfer through the PC
by agreeing on an updated split of the shared funds
which reflects the new balances after the transfer
and can be enforced on-chain anytime.
Special care is taken that participants can only ever
execute the most recent agreement on-chain.
PCs improve the throughput, latency and privacy of a blockchain system.

Payment networks (PNs)
\cite{%
poon_bitcoin_2016-1,%
miller_sprites_2017,%
dziembowski_perun:_2017,%
khalil_revive:_2017,%
dziembowski_general_2018-1%
}
can be constructed on top of PCs.
PCs can be
linked
to establish a path between a source and a destination
via some intermediaries. %
Hash- and time-locked contracts (HTLC)
are used to perform transfers via intermediaries
without counterparty risk,
giving rise to so called layer-2 (L2) PNs.
To this end, the destination draws a secret (preimage) and reveals a one-way function of the secret (preimage challenge) to the source.
The source then initiates a chain of payments to the destination,
all conditional on the revelation of
a valid preimage.
The destination reveals the secret to claim the funds, setting the chain of payments in motion.
Net,
the destination is paid, the source pays, and the intermediaries are in balance,
up to a small service fee earned for forwarding.
For recent surveys on PCs and PNs, see
\cite{%
jourenko_sok:_2019,%
gudgeon_sok:_2019%
}.

\subsection{Routing in Payment Networks}
\label{sec:routing-payment-networks}

Routing algorithms for PNs find paths and forward funds
while optimizing objectives such as
throughput,
latency,
or transaction fees.
Recently, the routing problem has been studied extensively
and various routing algorithms have been proposed
\cite{%
poon_bitcoin_2016-1,%
prihodko_flare:_2016,%
malavolta_silentwhispers:_2016,%
roos_settling_2017,%
hoenisch_aodvbased_2018,%
sivaraman_routing_2018,%
di_stasi_routing_2018,%
wang_flash:_2019%
}.
Early on,
it has been discovered,
also through the Lightning Torch experiment,
that single-path routing restricts the maximum transfer size.%
\footnote{%
\url{https://diar.co/volume-2-issue-25/\#1} (Jun 2018),
\url{https://www.coindesk.com/its-getting-harder-to-send-bitcoins-lightning-torch-heres-why} (Mar 2019)%
}
Multi-path routing \cite{piatkivskyi_split_2018}
allows for a more flexible use of PC liquidity
and hence can accommodate larger transfers
by splitting transfers into partial payments
that are routed along multiple paths.
However,
while single-path payments are naturally \emph{atomic},
\ie,
they either succeed  or fail entirely,
simply sending multiple partial payments
can lead to half-way incomplete transfers.
Atomicity is important %
to keep payment networks manageable from a systems-design perspective,
and is achieved by
atomic multi-path payments (AMP, \cite{osuntokun_[lightning-dev]_2018}),
which most state-of-the-art routing algorithms rely on
and Lightning developers
are actively working towards.%
\footnote{%
\url{https://bitcoin.stackexchange.com/q/89475} (Jul 2019)%
}
Yet,
because of its
`everyone-waits-for-the-last' philosophy,
atomicity comes at a cost
for latency and throughput.

\subsection{Main Contributions}
\label{sec:main-contributions}

\begin{figure}
    \centering
    \begin{tikzpicture}
        \begin{axis}[
                xlabel={Avg. Time-To-Completion\\for Successful Transfers [\si{\second}]},
                ylabel={Avg. Success Throughput\\per Node [\si{\xrp\per\second\per\nde}]},
                label style={
                    align=center,
                },
                xlabel style={
                    align=left,
                    at={(1.05,0)},
                    anchor=north west,
                    yshift={11.4pt},
                },
                width=0.45\linewidth,
                height=0.45\linewidth,
                legend style={
                    at={(1.1,0.95)},
                    anchor=north west,
                    draw=none,
                    font=\footnotesize,   %
                },
                legend cell align={left},
                label style={font=\footnotesize},   %
                tick label style={font=\footnotesize},   %
            ]

            \addplot+ [
                mark size=2,
                mark=+,
                only marks,
            ] table [
                x expr=(\thisrow{ttc_for_successful_tx-mean}),
                y expr=(\thisrow{throughput_success-mean}),
                discard if not={algo}{retry-02-amp-2},
            ] {figures/02_nodes100_txs500_paths25_edges4.605170to6.907755_paths25.data};
            \addlegendentry{Baseline};

            \addplot+ [
                mark size=1,
                mark=*,
                every mark/.append style={fill=red},
                only marks,
            ] table [
                x expr=(\thisrow{ttc_for_successful_tx-mean}),
                y expr=(\thisrow{throughput_success-mean}),
                discard if not={algo}{redundancy-02-amp-2},
            ] {figures/02_nodes100_txs500_paths25_edges4.605170to6.907755_paths25.data};
            \addlegendentry{Boomerang};

            \pgfplotsset{cycle list shift=-2};

            \addplot+ [
                name path=pareto_retry,
                no marks,
                thick,
            ] table [
                x expr=(\thisrow{ttc_for_successful_tx-mean}),
                y expr=(\thisrow{throughput_success-mean}),
                discard if not={algo}{retry-02-amp-2},
            ] {figures/pareto_retry-02-amp-2.data};

            \addplot+ [
                name path=pareto_redundancy,
                no marks,
                thick,
            ] table [
                x expr=(\thisrow{ttc_for_successful_tx-mean}),
                y expr=(\thisrow{throughput_success-mean}),
                discard if not={algo}{redundancy-02-amp-2},
            ] {figures/pareto_redundancy-02-amp-2.data};

            \pgfplotsset{cycle list shift=-2};

            \path[name path=limits_pareto_retry] (axis cs:0.5266116624312419,0) -- (axis cs:1,0) -- (axis cs:1,44.55598869554187);
            \path[name path=limits_pareto_redundancy] (axis cs:0.3014100356469785,0) -- (axis cs:1,0) -- (axis cs:1,89.71919803553533);

            \addplot+ [
                fill=blue,
                fill opacity=0.1,
            ] fill between [
                of=pareto_retry and limits_pareto_retry,
            ];

            \addplot+ [
                fill=red,
                fill opacity=0.1,
            ] fill between [
                of=pareto_redundancy and limits_pareto_redundancy,
            ];

        \end{axis}

        \ifthenelse{\equal{\jnPDFEXPORTVERSION}{springer}}{
            \draw [-Latex] (3.5,3) -- ++(-0.7,0.7) node [midway,above,xshift=0.2em,yshift=-0.2em,rotate=-45] {\scriptsize better};
        }{
            \draw [-Latex] (4.2,3.5) -- ++(-0.9,0.9) node [midway,above,xshift=0.2em,yshift=-0.2em,rotate=-45] {\footnotesize better};
        }
    \end{tikzpicture}
    \caption{
    The
    blue (resp.~red)
    points mark
    the tradeoff between
    latency and throughput
    of the baseline (resp.~Boomerang) AMP routing scheme,
    obtained by varying an internal parameter.
    The Pareto fronts (lines) and the achievable
    regions (shaded) of the tradeoff are shown.
    Boomerang yields
    a 2x increase in throughput and
    a \SI{40}{\percent} decrease in latency
    over the baseline scheme.
    }
    \label{fig:sim-results-main-results-combined}
\end{figure}
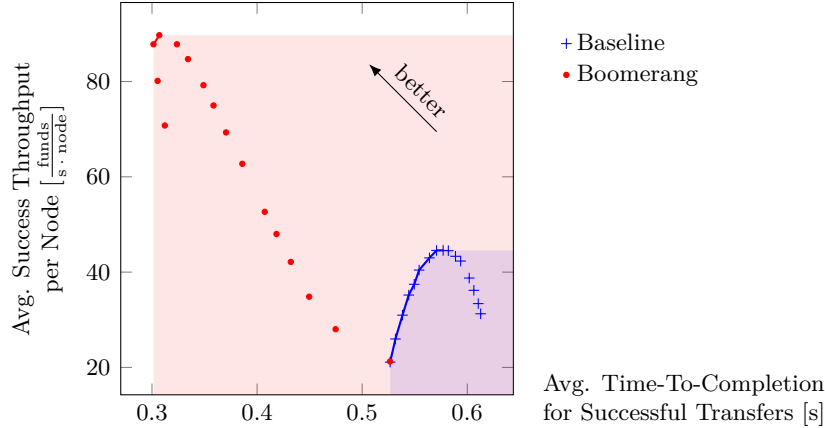

Due to the stochastic nature of PNs,
\ie,
random delays and failures of payment paths,
AMP routing often idles
while waiting for a few straggling paths.
This leads to a long time-to-completion (TTC) of transfers.
Furthermore,
already successful partial payments are kept pending,
seizing liquidity from the PN
that could have served
other transfers.
As a result, the throughput of the PN reduces.
For example,
consider a transfer
from Alice to Bob
of \$4 via 4 paths,
of which 3 succeed after \SI{1}{\second}
and 1 succeeds after \SI{4}{\second}.
The transfer has a TTC of \SI{4}{\second}
and consumes liquidity $\$4 \cdot \SI{4}{\second} = \$16 \cdot\si{\second}$.
Now
suppose instead
Alice could send 8 partial payments
of \$1 each (4 `extra' redundant paths),
such that
AMP completes once a quorum of 4 out of 8 paths succeeds
and Bob cannot steal any extra funds.
If
6 paths succeed after \SI{1}{\second}
and 2 succeed after \SI{4}{\second},
then
the transfer has a TTC of \SI{1}{\second}
and consumes liquidity $4 \cdot \$1 \cdot \SI{1}{\second} + 2 \cdot \$1 \cdot \SI{1}{\second} + 2 \cdot \$1 \cdot \SI{4}{\second} = \$14 \cdot\si{\second}$.
Thus,
the use of redundant payments
reduces the TTC
of AMP transfers.
As a result,
less liquidity is consumed
and the PN achieves
higher throughput.
Similar observations
about straggler mitigation using redundancy
have been made in large-scale distributed computing
\cite{dean_tail_2013,lee_speeding_2018,aktas_straggler_2019}.

We devise \emph{Boomerang},
a technique to be used on top of multi-path routing schemes to construct
redundant payment paths free of counterparty risk.
Building on ideas from publicly verifiable secret sharing,
we use a homomorphic one-way function
to intertwine the preimage challenges used for HTLC-type payment forwarding,
such that
Alice learns a secret of Bob
iff Bob overdraws funds from the redundant paths.
Funds are forwarded using Boomerang contracts,
which allow Alice to revert the transfer
iff she has learned Bob's secret.
We prove the Boomerang construction to be secure,
and present an implementation in Bitcoin Script
which applies
either adaptor signatures based on Schnorr or ECDSA signatures,
or elliptic curve scalar multiplication
as a one-way function.

We
empirically
verify
the %
benefits of Boomerang
for throughput and latency
of AMP routing.
For this purpose,
we choose a baseline routing scheme which resembles
the scheme currently used in Lightning.
We enhance this scheme with Boomerang.
As can be seen
in \figref{sim-results-main-results-combined},
redundancy
increases the throughput by 2x
and
reduces the TTC (latency) by \SI{40}{\percent}.
Our results suggest that redundancy is a generic tool
to boost the performance of
AMP routing algorithms.

\subsection{Paper Outline}
\label{sec:outline}

In \secref{preliminaries},
we recall cryptographic preliminaries,
introduce the system model,
and summarize causes for delay and failure of transfers in payment networks.
We devise the Boomerang construction for redundant transactions in \secref{boomerang-construction}, and prove it to be secure.
An implementation of Boomerang in Bitcoin Script is presented in \secref{implementation-bitcoin-script}.
In \secref{experimental-evaluation},
we demonstrate the utility of redundancy for improving routing protocols in experiments.

\section{Preliminaries}
\label{sec:preliminaries}

\subsection{System Model \& Terminology}
\label{sec:system-model}

The topology of a payment network (PN) is given by an undirected graph
whose nodes are agents, and whose edges are payment channels (PCs).
Each PC endpoint owns a share of the funds in the PC,
which we refer to as its liquidity or balance.
For privacy reasons it is undesirable to reveal balances
to third party nodes.
Following Lightning \cite[BOLT \#7]{poon_bitcoin_2016-1},
we assume that nodes have no information about PC balances when taking routing decisions,
but know the PN topology.
Finally,
we assume that PN nodes communicate in a peer-to-peer (P2P) gossiping fashion only along PCs,
\ie, PN topology is P2P network topology.

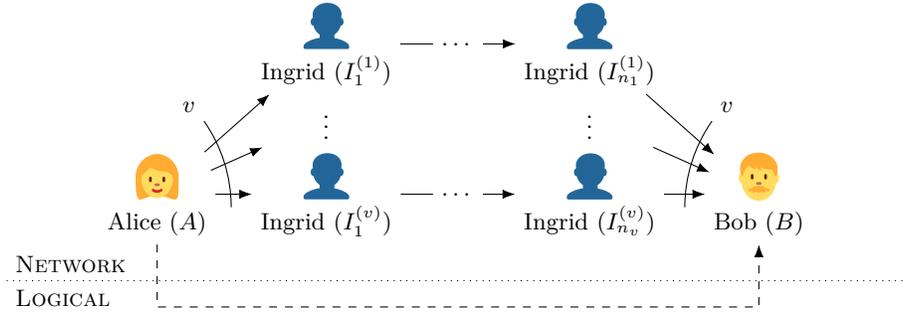
\begin{figure}[tb]
    \centering
    \begin{tikzpicture}
        \mathligsoff
        \ifthenelse{\equal{\jnPDFEXPORTVERSION}{springer}}{}{\footnotesize}
        
        \node [align=center] (a) at (-4,0) {{\Huge\emojialice{}}\\Alice ($A$)};
        \node [align=center] (b) at (+4,0) {{\Huge\emojibob{}}\\Bob ($B$)};
        
        \node [align=center] (i11) at (-1.75,2) {{\Huge\emojiingrid{}}\\Ingrid ($I_1^{(1)}$)};
        \node [align=center] (i1etc) at (0,2) {\dots};
        \node [align=center] (i1n) at (+1.75,2) {{\Huge\emojiingrid{}}\\Ingrid ($I_{n_1}^{(1)}$)};
        
        \node [align=center] (iV1) at (-1.75,0) {{\Huge\emojiingrid{}}\\Ingrid ($I_1^{(v)}$)};
        \node [align=center] (iVetc) at (0,0) {\dots};
        \node [align=center] (iVn) at (+1.75,0) {{\Huge\emojiingrid{}}\\Ingrid ($I_{n_v}^{(v)}$)};

        \node [align=center] (etc11) at (-1.75,+1) {\vdots};
        \node [align=center] (etc1n) at (+1.75,+1) {\vdots};

        \draw (i11) -- (i1etc);
        \draw [-Latex] (i1etc) -- (i1n);
        \draw [-Latex,shorten <=-0.5em] (a) -- (i11);
        \draw [-Latex,shorten >=-0.5em] (i1n) -- (b);
        \draw (iV1) -- (iVetc);
        \draw [-Latex] (iVetc) -- (iVn);
        \draw [-Latex] (a) -- (iV1);
        \draw [-Latex] (iVn) -- (b);

        \draw [-Latex,shorten >=2.3em,shorten <=-0.2em] (a) -- (etc11);
        \draw [-Latex,shorten <=2.3em,shorten >=-0.2em] (etc1n) -- (b);

        \draw (a) ++(-10:1) arc (0:35:2) node [above left] {$v$};
        \draw (b) ++(190:1) arc (180:145:2) node [above right] {$v$};

        \draw [-Latex,dashed] (a) -- ++(0,-1.5) -| (b);

        \draw [dotted] (-6,-1.15) -- (6,-1.15);
        \node [anchor=south west] at (-6,-1.15) {\textsc{Network}};
        \node [anchor=north west] at (-6,-1.15) {\textsc{Logical}};
        
    \end{tikzpicture}
    \caption{
    The TF from $A$ to $B$ (dashed) is implemented by $v$ TXs via intermediaries $(I_1^{(1)}, ..., I_{n_1}^{(1)})$, ..., $(I_1^{(v)}, ..., I_{n_v}^{(v)})$ (solid).
    }
    \label{fig:system-model}
\end{figure}

For multi-path routing, we highlight a strict separation between two layers in our terminology (\cf \figref{system-model}).
The `logical' layer
is concerned with \emph{transfers} (TFs) of an amount $v$ of funds
from a \emph{source} (Alice, $A$)
to a \emph{destination} (Bob, $B$).
The `network' layer
implements a TF through
multiple \emph{transactions} (TXs)
from the \emph{sender}
to the \emph{receiver}
along different paths
through multiple
intermediaries (Ingrid, $I_i$).
A preimage $p_i$
is used
for HTLC-type forwarding of the $i$-th TX.
The amounts of the TXs add up to the amount of the TF.
Routing algorithms
schedule a sequence of TXs in an attempt to implement a given TF.
Figures of merit are
\emph{throughput},
\ie, the average amount of funds transported successfully per time,
and
\emph{time-to-completion} (TTC),
\ie, the average delay between commencement and completion of a TF.

\subsection{Delay and Failure of Transactions in Payment Networks}

There are various causes for delay and failure of TXs in PNs:
\emph{a)} The liquidity of a PC along the desired path is insufficient.
\emph{b)} Insufficient fees do not incentivize intermediaries to forward.
\emph{c)} Queuing, propagation and processing delay of P2P messages.
\emph{d)} PN topology changes, nodes come and go, \eg, due to connectivity or maintenance.
\emph{e)} Governments or businesses attempt to censor certain TFs.

\subsection{Cryptographic Preliminaries}
\label{sec:cryptographic-preliminaries}

In this section, we briefly recapitulate the cryptographic tools used throughout the paper.
Let $\IG$ be a cyclic multiplicative group of prime order $q$ with a generator $g \in \IG$.
We assume that the
discrete logarithm problem (DLP, formally introduced in \secref{appendix-cryptographic-preliminaries})
is hard for $g$ in $\IG$,
which is commonly assumed to be the case,
\eg,
in certain elliptic curves (ECs)
used in Bitcoin.
Let $H\colon \IZ_q \to \IG$ with $H(x) := g^x$,
where $\IZ_q$ is the finite field of integers modulo $q$.
We require $H$ to be a one-way function,
which follows from
the DLP hardness assumption.
Given a \emph{preimage challenge} $h := H(x)$,
it is difficult to obtain the \emph{preimage} $x$,
but easy to check whether a purported preimage $\hat{x}$
satisfies the challenge (\ie, $H(\hat{x}) = h$).

$H$ has the following homomorphic property which we make extensive use of:
\begin{IEEEeqnarray}{rCl}
    \label{eq:homomorphic-H}
    \forall n \geq 1\colon
    \forall c_1, ..., c_n\colon
    \quad
    H\left( \sum_{i=1}^n c_i x_i \right) = \prod_{i=1}^n H(x_i)^{c_i}
\end{IEEEeqnarray}

\section{The Boomerang Construction}
\label{sec:boomerang-construction}

In this section, we show how to add redundant TXs without counterparty risk.
To this end, we build up the Boomerang construction in three steps.
First,
$B$ draws a secret $\alpha_0$.
We use the homomorphic property of $H$
and
ideas from
publicly verifiable secret sharing
\cite{shamir_how_1979,cohen_benaloh_secret_1987,feldman_practical_1987,pedersen_non-interactive_1992,schoenmakers_simple_1999}
to construct
the preimages $p_i$
and
the preimage challenges $H(p_i)$
such that
$A$ learns $\alpha_0$
iff $B$ overdraws funds from the redundant paths.
Second,
the so called Boomerang contract
serves as a building block
for contingent transfer of funds,
\ie,
HTLC-type forwarding with the additional provision
that $A$ can revert the TF
iff she learns $\alpha_0$.
Finally,
the end-to-end procedure of a Boomerang TF
is devised from the previously mentioned ingredients.
Subsequently, we prove that Boomerang
satisfies the relevant notions of security.

\subsection{Setup of Preimage Challenges}
\label{sec:setup-preimage-challenges}

\begin{figure}
    \centering
    \begin{tikzpicture}
        \mathligsoff
        \ifthenelse{\equal{\jnPDFEXPORTVERSION}{springer}}{}{\footnotesize}
        
        \node [align=center] (a) at (-4,0) {{\Huge\emojialice{}}\\Alice ($A$)};
        \node [align=center] (b) at (+4,0) {{\Huge\emojibob{}}\\Bob ($B$)};
        
        \node [align=center] (i11) at (-1.75,2) {{\Huge\emojiingrid{}}\\Ingrid ($I_1^{(1)}$)};
        \node [align=center] (i1etc) at (0,2) {\dots};
        \node [align=center] (i1n) at (+1.75,2) {{\Huge\emojiingrid{}}\\Ingrid ($I_{n_1}^{(1)}$)};
        
        \node [align=center] (iV1) at (-1.75,0) {{\Huge\emojiingrid{}}\\Ingrid ($I_1^{(v)}$)};
        \node [align=center] (iVetc) at (0,0) {\dots};
        \node [align=center] (iVn) at (+1.75,0) {{\Huge\emojiingrid{}}\\Ingrid ($I_{n_v}^{(v)}$)};
        
        \node [red,align=center] (iW1) at (-1.75,-2) {{\Huge\emojiingrid{}}\\Ingrid ($I_1^{(w)}$)};
        \node [red,align=center] (iWetc) at (0,-2) {\dots};
        \node [red,align=center] (iWn) at (+1.75,-2) {{\Huge\emojiingrid{}}\\Ingrid ($I_{n_w}^{(w)}$)};
        
        \node [align=center] (etc11) at (-1.75,+1) {\vdots};
        \node [align=center] (etc1n) at (+1.75,+1) {\vdots};
        \node [red,align=center] (etc21) at (-1.75,-1) {\vdots};
        \node [red,align=center] (etc2n) at (+1.75,-1) {\vdots};
        
        \draw (i11) -- (i1etc);
        \draw [-Latex] (i1etc) -- (i1n);
        \draw [-Latex,shorten <=-0.5em] (a) -- (i11);
        \draw [-Latex,shorten >=-0.5em] (i1n) -- (b);
        \draw (iV1) -- (iVetc);
        \draw [-Latex] (iVetc) -- (iVn);
        \draw [-Latex] (a) -- (iV1);
        \draw [-Latex] (iVn) -- (b);
        \draw [red] (iW1) -- (iWetc);
        \draw [red,-Latex] (iWetc) -- (iWn);
        \draw [red,-Latex] (a) -- (iW1);
        \draw [red,-Latex] (iWn) -- (b);
        
        \draw [-Latex,shorten >=2.3em,shorten <=-0.2em] (a) -- (etc11);
        \draw [-Latex,shorten <=2.3em,shorten >=-0.2em] (etc1n) -- (b);
        \draw [red,-Latex,shorten >=2.3em] (a) -- (etc21);
        \draw [red,-Latex,shorten <=2.3em] (etc2n) -- (b);
        
        \draw (a) ++(-10:1) arc (0:35:2) node [above left] {$v$};
        \draw (b) ++(190:1) arc (180:145:2) node [above right] {$v$};
        
        \draw [red] (a) ++(-15:1.15) arc (0:-35:1.5) node [below left] {$u$};
        \draw [red] (b) ++(195:1.15) arc (180:215:1.5) node [below right] {$u$};
        
    \end{tikzpicture}
    \caption{
    The TF of amount $v$ from $A$ to $B$ is implemented by the first $v$ out of $w:=v+u$ TXs of unit amount \$1 via $(I_1^{(1)}, ..., I_{n_1}^{(1)})$, ..., $(I_1^{(w)}, ..., I_{n_w}^{(w)})$.
    }
    \label{fig:system-model-w-redundancy}
\end{figure}
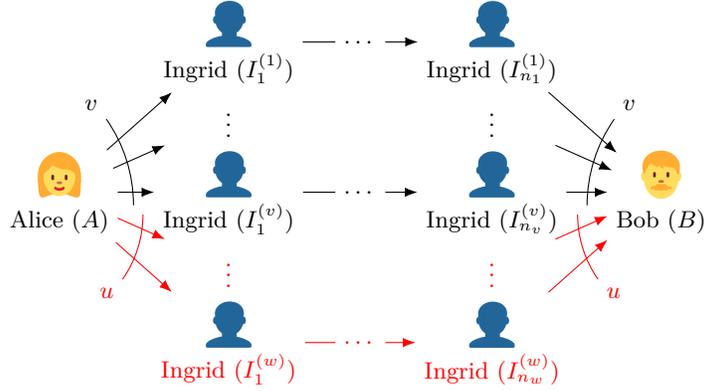

We assume that $A$ and $B$ have agreed out-of-band to partition their TF into $v$ TXs of a unit of funds (\$1) each, without loss of generality (w.l.o.g.).
In addition, they use $u$ redundant TXs to improve their TF,
so that the total number of TXs is $w := v + u$ (\cf \figref{system-model-w-redundancy}).
First,
$B$ chooses a polynomial $P(x)$ of degree $\operatorname{deg}(P) = v$ with coefficients $\alpha_0 \drawrandom \IZ_q, ..., \alpha_v \drawrandom \IZ_q$ drawn uniformly at random,
\begin{IEEEeqnarray}{rCl}
    \label{eq:polynomial}
    P(x) &=& \sum_{j=0}^{v} \alpha_j x^j.
\end{IEEEeqnarray}
Then, $B$ commits to $P(x)$ by providing $A$ with $H(\alpha_0), ..., H(\alpha_v)$.
Due to the homomorphic property of $H$,
\eqref{homomorphic-H},
$A$ can compute
\begin{IEEEeqnarray}{rCl}
    \label{eq:preimage-challenges-intertwining}
    \forall i \in \{1,...,v,...,w\}\colon
    \quad
        H(P(i))
        = H\left( \sum_{j=0}^{v} \alpha_j i^j \right)
        &=& \prod_{j=0}^v H(\alpha_j)^{\left(i^j\right)}.
\end{IEEEeqnarray}
For the $i$-th TX,
$p_i := P(i)$ is used as a preimage for HTLC-type forwarding.
Hence, $A$ uses $H(p_i) = H(P(i))$ from \eqref{preimage-challenges-intertwining}
as preimage challenge,
and informs $B$ out-of-band of which $i$ was used.

To redeem the $i$-th TX, $B$ reveals $p_i = P(i)$.
Should $B$ overdraw by revealing more than $v$ evaluations $p_i$ of $P(x)$,
then $\alpha_0, ..., \alpha_v$ can be recovered using polynomial interpolation due to $\operatorname{deg}(P) = v$.
Recall that $\alpha_0$ serves as a secret
which $A$ can use to revert the TF in this case.
As long as $B$ reveals no more than $v$ evaluations of $P(x)$,
each $\alpha_i$ remains marginally uniformly distributed.
In this case, the TF is final.

Note that $A$ and $B$ do not have to agree on $u$ ahead of time.
Instead,
$A$ can create virtually infinitely many $H(p_i)$ (as long as $w < q$),
and hence send a continuous flow of TXs until $B$ redeems $v$ TXs.
This property is similar to rateless codes used to
implement `digital fountains' \cite{DBLP:conf/sigcomm/ByersLMR98},
and enables $B$ to choose $u$ adaptively during execution of the TF. %
Furthermore,
since there is no risk in `losing control' over a redundant TX,
source routing can be abandoned
in favor of the PN taking distributed routing decisions.
Conceptually, the use of redundant TXs in PNs is analogous to
the use of erasure-correcting codes \cite{paper_Elias1956} in
communication networks consisting of packet erasure channels
\cite{DBLP:journals/rfc/rfc6330}
and for straggler mitigation in large-scale distributed computing and storage
\cite{dean_tail_2013,lee_speeding_2018}.

\subsection{Boomerang Contract}
\label{sec:boomerang-contract}

We devise the Boomerang contract,
which
implements reversible HTLC-type forwarding
as required for Boomerang
on top of a PC.
Still,
\wolog
\$1 funds and $\delta$ TX fee are to be forwarded from party $P_1$ to party $P_2$.

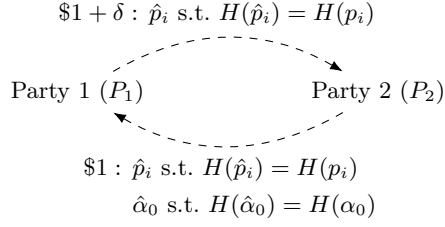
\begin{figure}
    \centering
    \begin{tikzpicture}
        \mathligsoff
        \ifthenelse{\equal{\jnPDFEXPORTVERSION}{springer}}{}{\footnotesize}
        
        \node (p1) at (-2,0) {Party $1$ ($P_1$)};
        \node (p2) at (+2,0) {Party $2$ ($P_2$)};
        \draw [-Latex,txPromised] (p1) to [bend left] node [midway,above] {
            $\begin{aligned}
            \$1+\delta: \:\:& \hat{p}_i \text{ s.t.\ } H(\hat{p}_i) = H(p_i)
            \end{aligned}$
            } (p2);
        \draw [-Latex,txPromised] (p2) to [bend left] node [midway,below] {
            $\begin{aligned}
            \$1: \:\:& \hat{p}_i \text{ s.t.\ } H(\hat{p}_i) = H(p_i)   \\
                & \hat{\alpha}_0 \text{ s.t.\ } H(\hat{\alpha}_0) = H(\alpha_0)
            \end{aligned}$
            } (p1);
        
    \end{tikzpicture}
    \caption{Boomerang contract used to forward \$1 plus $\delta$ TX fees from $P_1$ to $P_2$}
    \label{fig:boomerang-contract}
\end{figure}

Conceptually, the desired behavior could be accomplished by two conditional forwardings (\cf \figref{boomerang-contract}),
the first of which (`forward', top arrow in \figref{boomerang-contract})
transfers $\$1+\delta$ from $P_1$ to $P_2$ upon revelation of a preimage $\hat{p}_i$ such that
$H(\hat{p}_i) = H(p_i)$,
and the second of which (`reverse', bottom arrow in \figref{boomerang-contract})
transfers $\$1$ back from $P_2$ to $P_1$ upon revelation of two preimages $\hat{p}_i$ and $\hat{\alpha}_0$ such that $H(\hat{p}_i) = H(p_i)$ and $H(\hat{\alpha}_0) = H(\alpha_0)$.
Note that the two conditions are nested such that
if $P_1$ redeems the second forwarding,
then $P_2$ can redeem the first forwarding.
As a result, the Boomerang contract has three possible outcomes.
Either,
\emph{a)} neither forwarding is redeemed and $P_1$ retains all funds (\eg, in the case of a timeout or an unused redundant TX),
or
\emph{b)} only the `forward' forwarding is redeemed (\ie, $B$ draws funds by revealing $P(i)$ but does not leak $\alpha_0$),
or
\emph{c)} both forwardings are redeemed (\ie, $B$ overdraws and reveals both $P(i)$ and $\alpha_0$).
Anyway, %
$P_2$ cannot loose funds and thus agrees to deploy the contract on the PC.

\begin{figure}%
    \begin{subfigure}[t]{0.48\linewidth}
        \centering
        \begin{tikzpicture}
            \mathligsoff
            \ifthenelse{\equal{\jnPDFEXPORTVERSION}{springer}}{}{\footnotesize}
            
            \node [fcFunds] (in) at (0,0) {\$1};
            \node [fcDecision,node distance=0.5cm and 1cm,below=of in] (decision1) {$P_2$ reveals $\hat{p}_i$ s.t.\ $H(\hat{p}_i) = H(p_i)$ \\ within $t\in[T_0, T_0{+}\Delta_{\mathrm{fwd}}]$};
            \node [fcDecision,below=of decision1] (decision2) {$P_1$ reveals $\hat{\alpha}_0$ s.t.\ $H(\hat{\alpha}_0) = H(\alpha_0)$ \\ within $t\in[T_0, T_0{+}\Delta_{\mathrm{rev}}]$};
            \node [fcPayout,right=of decision1] (payout1) {$P_1$};
            \node [fcPayout,right=of decision2] (payout2) {$P_2$};
            \node [fcPayout,below=of decision2] (payout3) {$P_1$};
            
            \draw [fcLine] (in) -- (decision1);
            \draw [fcLine] (decision1) -- (decision2) node [midway,right] {yes};
            \draw [fcLine] (decision1) -- (payout1) node [midway,above] {no};
            \draw [fcLine] (decision2) -- (payout2) node [midway,above] {no};
            \draw [fcLine] (decision2) -- (payout3) node [midway,right] {yes};
            
        \end{tikzpicture}
        \caption{}
        \label{fig:flow-chart-funds}
    \end{subfigure}
    \hfill
    \begin{subfigure}[t]{0.48\linewidth}
        \centering
        \begin{tikzpicture}
            \mathligsoff
            \ifthenelse{\equal{\jnPDFEXPORTVERSION}{springer}}{}{\footnotesize}
            
            \node [fcFunds] (in) at (0,0) {$\delta$};
            \node [fcDecision,below=of in] (decision1) {$P_2$ reveals $\hat{p}_i$ s.t.\ $H(\hat{p}_i) = H(p_i)$ \\ within $t\in[T_0, T_0{+}\Delta_{\mathrm{fwd}}]$};
            \node [fcPayout,right=of decision1] (payout1) {$P_1$};
            \node [fcPayout,below=of decision1] (payout2) {$P_2$};
            
            \draw [fcLine] (in) -- (decision1);
            \draw [fcLine] (decision1) -- (payout1) node [midway,above] {no};
            \draw [fcLine] (decision1) -- (payout2) node [midway,right] {yes};
            
        \end{tikzpicture}
        \caption{}
        \label{fig:flow-chart-txfee}
    \end{subfigure}
    \caption{Flow charts for payout of a Boomerang contract (\cf \figref{boomerang-contract}) concerning \subref{fig:flow-chart-funds} \$1 funds and \subref{fig:flow-chart-txfee} $\delta$ TX fees between $P_1$ and $P_2$}
    \label{fig:flow-charts}
\end{figure}
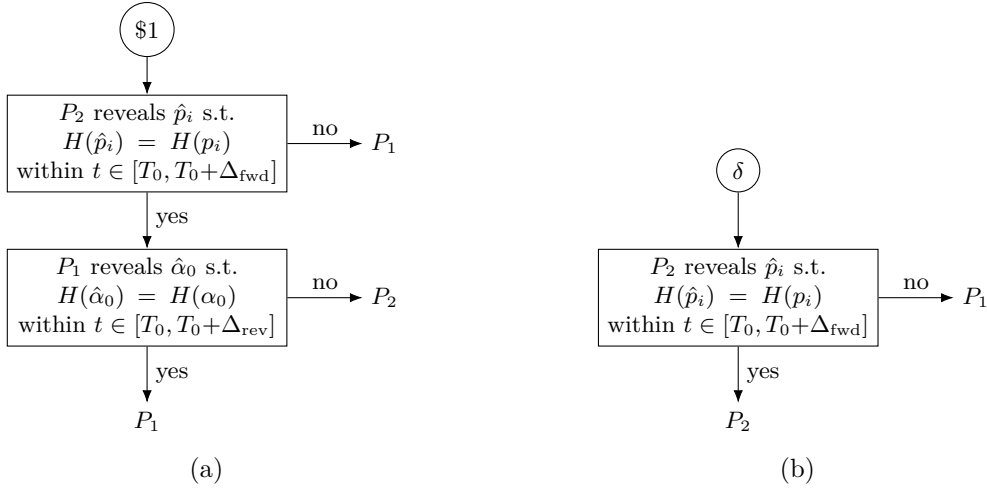

Note that \figref{boomerang-contract} has been simplified for ease of exposition.
Indeed, two essential aspects are not captured in \figref{boomerang-contract} and necessitate the refined final specification of the Boomerang contract in \figref{flow-charts}.
First,
timeouts $\Delta_{\mathrm{fwd}}$ and $\Delta_{\mathrm{rev}}$
(relative to current time $T_0$)
need to be chosen such that the source of the TF has time
to detect overdraw and reclaim the funds ($\Delta_{\mathrm{fwd}} < \Delta_{\mathrm{rev}}$).
Second,
having two separate forwardings would requisition twice the
liquidity
($\$2+\delta$)
while the forwardings are pending.
Instead,
the Boomerang contract as specified by the flow charts in \figref{flow-charts}
allows for consistent timeouts with $\Delta_{\mathrm{fwd}} < \Delta_{\mathrm{rev}}$
and
requisitions the TX amount in PC liquidity only once
($\$1+\delta$).
We present an implementation of \figref{flow-charts}
in Bitcoin Script
on top of Eltoo PCs \cite{decker_eltoo:_2018}
in \secref{implementation-bitcoin-script}.
Throughout the paper,
we use two circular arrows
as in \figref{boomerang-contract}
to visualize a Boomerang contract.

\subsection{End-To-End Procedure for Boomerang Transfers}
\label{sec:end-to-end-procedure-boomerang-transfer}

Given the intertwined preimage challenges of \secref{setup-preimage-challenges}
and the Boomerang contract of \secref{boomerang-contract},
we devise the end-to-end procedure for Boomerang TFs.

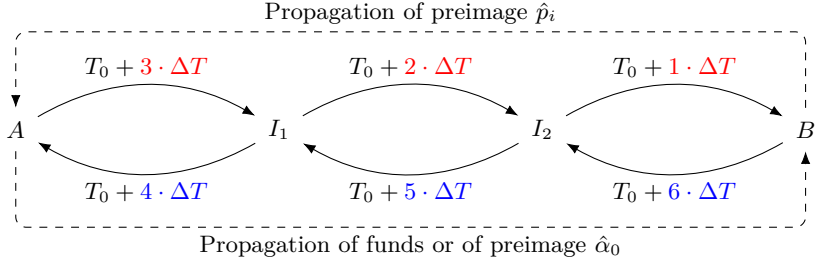
\begin{figure}[tb!]
    \centering
    \begin{tikzpicture}[
            every path/.append style={
                shorten <=0.2em,
                shorten >=0.2em,
            },
        ]
        \mathligsoff
        \ifthenelse{\equal{\jnPDFEXPORTVERSION}{springer}}{}{\footnotesize}
        
        \node [] (A) at (0,0) {$A$};
        \node [] (I1) at (3.5,0) {$I_1$};
        \node [] (I2) at (7,0) {$I_2$};
        \node [] (B) at (10.5,0) {$B$};
        
        \draw [-Latex] (A) to [bend left] node [midway,above] {$T_0 + \textcolor{red}{3 \cdot \Delta T}$} (I1);
        \draw [-Latex] (I1) to [bend left] node [midway,above] {$T_0 + \textcolor{red}{2 \cdot \Delta T}$} (I2);
        \draw [-Latex] (I2) to [bend left] node [midway,above] {$T_0 + \textcolor{red}{1 \cdot \Delta T}$} (B);

        \draw [-Latex] (B) to [bend left] node [midway,below] {$T_0 + \textcolor{blue}{6 \cdot \Delta T}$} (I2);
        \draw [-Latex] (I2) to [bend left] node [midway,below] {$T_0 + \textcolor{blue}{5 \cdot \Delta T}$} (I1);
        \draw [-Latex] (I1) to [bend left] node [midway,below] {$T_0 + \textcolor{blue}{4 \cdot \Delta T}$} (A);
        
        \draw [dashed,-Latex,rounded corners] (B) -- ++(0,1.3) -| (A) node [pos=0.25,above] {Propagation of preimage $\hat{p}_i$};
        \draw [dashed,-Latex,rounded corners] (A) -- ++(0,-1.3) -| (B) node [pos=0.25,below] {Propagation of funds or of preimage $\hat{\alpha}_0$};
        
    \end{tikzpicture}
    \caption[]{
        Staggering of timeouts of Boomerang contracts:
        The \textcolor{red}{$\Delta_{\mathrm{fwd}}$} are chosen such that propagation of the preimage $\hat{p}_i$ from $B$ via $I_2$ and $I_1$ to $A$ is guaranteed once $B$ reveals $\hat{p}_i$.
        The \textcolor{blue}{$\Delta_{\mathrm{rev}}$} are chosen such that propagation of the preimage $\hat{\alpha}_0$ from $A$ via $I_1$ and $I_2$ to $B$ is guaranteed once $A$ reveals $\hat{\alpha}_0$.
    }
    \label{fig:timeouts-forwarding}
\end{figure}

To forward a TX along a path from $A$ to $B$ via $I_1, ..., I_n$,
the timeouts $\Delta_{\mathrm{fwd}}$ and $\Delta_{\mathrm{rev}}$
of the Boomerang contracts between $A$ and $I_1$, $I_i$ and $I_{i+1}$, and $I_n$ and $B$
need to be chosen in the following way:
The forward components' timeouts decrease along the path from $A$ to $B$.
The reverse components' timeouts increase along the path from $A$ to $B$.
Additionally, the earliest reverse component expires later than the latest forward component,
to allow $A$ to react to overdraw and activate the reverse components.
An example is illustrated in \figref{timeouts-forwarding}.

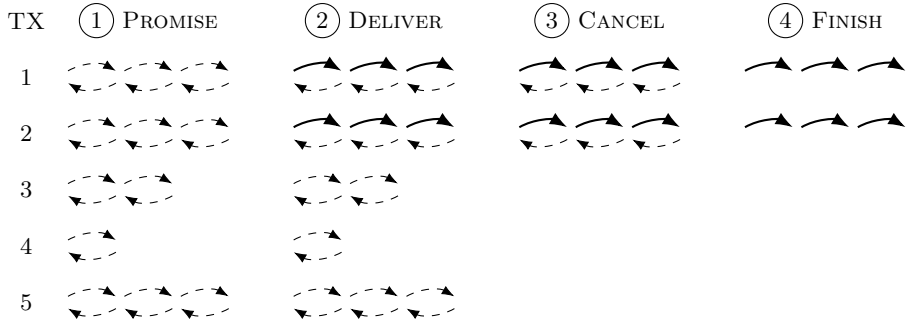
\begin{figure}%
    \centering
    \begin{tikzpicture}[
            every path/.append style={
                shorten <=0.2em,
                shorten >=0.2em,
            },
        ]
        \mathligsoff
        \ifthenelse{\equal{\jnPDFEXPORTVERSION}{springer}}{}{\footnotesize}
        
        \begin{scope}[xshift=-0.5cm,scale=0.75]
            \node at (0,1) {\textsc{TX}};
            \foreach \j in {1,...,5} {
                \node at (0,-\j+1) {$\j$};
            }
        \end{scope}
        
        \begin{scope}[xshift=0cm,scale=0.75]
            \node at (1.5cm,1) {\circled{1} \textsc{Promise}};
            \foreach \j/\d/\sF/\sB in {0/2/txPromised/txPromised,1/2/txPromised/txPromised,2/1/txPromised/txPromised,3/0/txPromised/txPromised,4/2/txPromised/txPromised} {
                \foreach \i in {0,...,\d} {
                    \draw [\sF,-Latex,yshift=+0.2em] (\i,-\j) to[bend left] (\i+1,-\j);
                    \draw [\sB,-Latex,yshift=-0.2em] (\i+1,-\j) to[bend left] (\i,-\j);
                }
            }
        \end{scope}
        
        \begin{scope}[xshift=3cm,scale=0.75]
            \node at (1.5cm,1) {\circled{2} \textsc{Deliver}};
            \foreach \j/\d/\sF/\sB in {0/2/txDelivered/txPromised,1/2/txDelivered/txPromised,2/1/txPromised/txPromised,3/0/txPromised/txPromised,4/2/txPromised/txPromised} {
                \foreach \i in {0,...,\d} {
                    \draw [\sF,-Latex,yshift=+0.2em] (\i,-\j) to[bend left] (\i+1,-\j);
                    \draw [\sB,-Latex,yshift=-0.2em] (\i+1,-\j) to[bend left] (\i,-\j);
                }
            }
        \end{scope}
        
        \begin{scope}[xshift=6cm,scale=0.75]
            \node at (1.5cm,1) {\circled{3} \textsc{Cancel}};
            \foreach \j/\d/\sF/\sB in {0/2/txDelivered/txPromised,1/2/txDelivered/txPromised} {
                \foreach \i in {0,...,\d} {
                    \draw [\sF,-Latex,yshift=+0.2em] (\i,-\j) to[bend left] (\i+1,-\j);
                    \draw [\sB,-Latex,yshift=-0.2em] (\i+1,-\j) to[bend left] (\i,-\j);
                }
            }
        \end{scope}
        
        \begin{scope}[xshift=9cm,scale=0.75]
            \node at (1.5cm,1) {\circled{4} \textsc{Finish}};
            \foreach \j/\d/\sF/\sB in {0/2/txDelivered/txNone,1/2/txDelivered/txNone} {
                \foreach \i in {0,...,\d} {
                    \draw [\sF,-Latex,yshift=+0.2em] (\i,-\j) to[bend left] (\i+1,-\j);
                    \draw [\sB,-Latex,yshift=-0.2em] (\i+1,-\j) to[bend left] (\i,-\j);
                }
            }
        \end{scope}
        
    \end{tikzpicture}
    \caption[]{
        Stages of a Boomerang TF
        (here using $w=v+u=2+3$ TXs,
        each through two $I_i$):
        \circled{1}
        $A$ attempts $5$ TXs;
        TXs $1$, $2$ and $5$ reach $B$,
        $3$ and $4$ do not.
        \circled{2}
        $B$ claims TXs $1$ and $2$ by revealing $P(1)$ and $P(2)$.
        \circled{3}
        The unsuccessful and outstanding TXs $3$, $4$ and $5$ are cancelled upon request from $B$.
        \circled{4}
        $A$ relinquishes the option to retract TXs $1$ and $2$ as no further funds can be drawn by $B$.
    }
    \label{fig:transfer-dynamic-behavior}
\end{figure}

A Boomerang TF of $v$ units using $w = v + u$ redundant TXs proceeds, after the setup of preimage challenges out-of-band, in four steps (\cf \figref{transfer-dynamic-behavior}):
\begin{enumerate}
    \item[\circled{1}] \textsc{Promise}:
        $A$ attempts $w = v+u$ TXs.
    \item[\circled{2}] \textsc{Deliver}:
        $B$ claims funds from up to $v$ TXs by revealing the preimages $p_i = P(i)$ to the corresponding preimage challenges $H(p_i)$, activating the forward components of the Boomerang contracts along the respective paths.
    \item[\circled{3}] \textsc{Cancel}:
        Upon a request by $B$ which is passed on to the tip of the path of an unsuccessful or surplus TX, outstanding TXs are cancelled.
        Note that cancellation is counterparty risk-free if it proceeds along the path from $B$ towards $A$,
        and honest as well as rational participants have a self-interest in freeing up liquidity that will foreseeably not earn TX fees.
    \item[\circled{4}] \textsc{Finish}:
        $A$ renounces the reverse components of the remaining Boomerang contracts to free up liquidity, as there is no more risk of $B$ overdrawing.
\end{enumerate}

\subsection{Security Guarantees}
\label{sec:security-guarantees}

We give the following guarantees to the source $A$ and the destination $B$ of a TF,
respectively,
which formalize security for the outlined Boomerang construction:
\begin{theorem}[$A$-Guarantee]
    \label{thm:a-guarantee}
    If more than $v$ of the redundant TXs are drawn from $A$,
    then $A$ can recover $\alpha_0$ and revert all TXs.
\end{theorem}
\begin{theorem}[$B$-Guarantee]
    \label{thm:b-guarantee}
    As long as $B$ follows the protocol
    and draws no more than $v$ of the redundant TXs,
    all TXs are final,
    except with negligible probability,
    provided the DLP is hard for $(\IG, g)$.
\end{theorem}
\begin{corollary}[Proofs]
    \label{thm:proofs}
    If $A$ knows $\alpha_0$, this proves that $B$ cheated.
    If $A$ knows $p_i$, this proves that $B$ was paid accordingly.
    $A$ can forge the proofs only with negligible probability,
    provided the DLP is hard for $(\IG, g)$.
\end{corollary}
Note that `drawing' a TX means `revealing the preimage $p_i = P(i)$ for the challenge $H(p_i)$ of TX $i$'.
Hence, the above guarantees
are statements about `how much preimage information flows' `into $A$' and `out of $B$',
implicitely assuming the worst-case that all intermediary $I_i$ collude with the respective opposing party.

\begin{proof}[Proof of Theorem~\ref{thm:a-guarantee}: $A$-Guarantee]
Due to the homomorphic property of $H$,
and the fact that $A$ computes the challenges $H(p_i)$
from the commitments $H(\alpha_0), ..., H(\alpha_v)$,
there exists a unique polynomial of degree $v$
that passes through all $P(i)$.
This polynomial can be determined by interpolation,
if more than $v$ preimages $p_i = P(i)$ are revealed.
Hence,
in the case of overdrawing,
$A$ can be certain to obtain an $\alpha_0$,
such that $A$ can activate the reverse components of the Boomerang contracts
and revert all TXs.
\ifthenelse{\equal{\jnPDFEXPORTVERSION}{springer}}{\qed}{}
\end{proof}

\begin{proof}[Proof of Theorem~\ref{thm:b-guarantee}: $B$-Guarantee]
It suffices to show that
one cannot recover $\alpha_0$
from $H(\alpha_0)$, ..., $H(\alpha_v)$, $P(i_1)$, ..., $P(i_v)$,
with $i_1, ..., i_v$ distinct,
except with negligible probability,
provided that the DLP is hard for $(\IG, g)$.

W.l.o.g., assume that all preimages revealed by $B$
have been forwarded to $A$.
Call the task $A$ faces \emph{Alice's problem} (AP),
\ie,
given an AP instance $(g^{\alpha_0}, ..., g^{\alpha_v}, P(i_1), ..., P(i_v))$
for known distinct $i_1, ..., i_v$,
find $\alpha_0$.
We show that AP is computationally infeasible
by reducing the DLP to AP.

The reduction proceeds as follows.
Assume $\CB$ is a blackbox that solves AP instances efficiently.
Using $\CB$, we construct a DLP solver $\CA$.
Present $\CA$ with a DLP instance $(g, g^a)$.
$\CA$ samples $P(i_j) \drawrandom \IZ_q$ for $j = 1,...,v$ uniformly at random,
such that there exists a unique degree $v$ polynomial interpolating $P(0), P(i_1), ..., P(i_v)$,
with coefficients $\alpha_0, ..., \alpha_v$ chosen uniformly at random.
Recall that $P(0) = \alpha_0 = a$ is unknown.
By \eqref{polynomial},
\begin{IEEEeqnarray}{C}
    \underbrace{
    \begin{bmatrix}
        P(0)   \\
        P(i_1)   \\
        \vdots   \\
        P(i_v)
    \end{bmatrix}
    }_{
    :=\,
    \Bp
    }
    =
    \underbrace{
    \begin{bmatrix}
        0^0         & 0^1         & \dots       & 0^v     \\
        i_1^0       & i_1^1       & \dots       & i_1^v   \\
        \vdots      & \vdots      &             & \vdots  \\
        i_v^0       & i_v^1       & \dots       & i_v^v  
    \end{bmatrix}
    }_{
    :=\,
    \BM
    }
    \underbrace{
    \begin{bmatrix}
        \alpha_0   \\
        \alpha_1   \\
        \vdots   \\
        \alpha_v
    \end{bmatrix}
    }_{
    :=\,
    \Balpha
    },
\end{IEEEeqnarray}
where the entries of $\BM$ are fixed. %

As $\BM$ is a Vandermonde matrix and thus invertible,
\begin{IEEEeqnarray}{C}
    \Bp = \BM \Balpha
    \quad
    \iff
    \quad
    \Balpha = \BM^{-1} \Bp.
\end{IEEEeqnarray}
Let $\tilde{m}_{ij} := \left\{ \BM^{-1} \right\}_{ij}$
and $p_j := \left\{ \Bp \right\}_j$,
where $\{X\}_y$ is the $y$-th entry of $X$.
Then,
\begin{IEEEeqnarray}{C}
    \alpha_i = \sum_{j=0}^v \tilde{m}_{ij} p_j.
\end{IEEEeqnarray}
Hence, using the homomorphic property of $H$ from \eqref{homomorphic-H}, obtain
\begin{IEEEeqnarray}{C}
    \forall i\in\{0, ..., v\}\colon
    \quad
    H(\alpha_i) = H\left( \sum_{j=0}^v \tilde{m}_{ij} p_j \right) = \prod_{j=0}^v H(p_j)^{\tilde{m}_{ij}}.
\end{IEEEeqnarray}

Now, $\CA$ has a tuple
$(H(\alpha_0), ..., H(\alpha_v), P(i_1), ..., P(i_v))$
drawn from the distribution of AP instances implied by the Boomerang protocol.
$\CA$ invokes the AP oracle $\CB$
and obtains $\alpha_0 = a$,
which solves the DLP instance.

Thus, an efficient solution to AP implies an efficient solution to DLP.
But since DLP is assumed to be hard,
AP has to be hard, proving the claim.
\ifthenelse{\equal{\jnPDFEXPORTVERSION}{springer}}{\qed}{}
\end{proof}

\begin{proof}[Proof of Corollary~\ref{thm:proofs}: Proofs]
Theorems \ref{thm:a-guarantee} and \ref{thm:b-guarantee} imply
that
$A$ knows $\alpha_0$ iff $B$ cheated
(except with negligible probability, provided the DLP is hard for $(\IG, g)$).

For $p_i$ to serve as proof of payment,
it requires that
for any $n <= v$
and
known distinct $i_1, ..., i_{n+1}$,
$A$
cannot obtain $P(i_{n+1})$ from
$H(\alpha_0)$, ..., $H(\alpha_v)$, $P(i_1)$, ..., $P(i_n)$
(except with negligible probability, provided the DLP is hard for $(\IG, g)$).
For $n=v$,
this follows from Theorem~\ref{thm:b-guarantee}
by contradiction,
as $\alpha_0$ can be computed from $P(i_1), ..., P(i_{v+1})$.
For $n<v$,
the problem is harder than for $n=v$.
\ifthenelse{\equal{\jnPDFEXPORTVERSION}{springer}}{\qed}{}
\end{proof}

\section{Implementation in Bitcoin Script}
\label{sec:implementation-bitcoin-script}

We present an implementation of Boomerang in Bitcoin Script
and its deployment on an Eltoo PC.
For simplicity,
we first present an implementation which hinges on a new opcode for $H(x)$.
Subsequently,
we refine the implementation to
apply adaptor signatures based on Schnorr or ECDSA signatures instead.

\subsection{Implementation Using an Opcode for $H(x)$}
\label{sec:implementation-opcode-H}

Assume the addition of a new command \verb|ECEXP|$\langle{}g\rangle{}$ to Bitcoin Script
for computing $H(x) = g^x \in \IG$.
The command would pop $x$ off the stack and push $g^x$ back onto the stack.
The necessary cryptographic primitives to do this in ECs
are already part of Bitcoin and only need to be exposed to the scripting engine.
Then,
\verb|ECEXP|$\langle{}g\rangle{}$
can be used as a one-way function in similar situations as the
\verb|SHA*|
or
\verb|HASH*|
commands,
but provides the homomorphic property
that can be necessary or useful for applications beyond Boomerang.

We show how the mechanics of the Boomerang contract
specified
in \figref{flow-charts}
can be implemented using Bitcoin Script.
To this end,
we refer to the settlement transaction of the PC on top of which the Boomerang contract is deployed as $\mathsf{TX}_{\mathrm{settle}}$,
and to transactions that $P_i$ would use to pay a certain output to themselves as
$\mathsf{TX}_{\mathrm{payout},P_i}$.
For the signature scheme employed by Bitcoin,
an identity $P$ has a public key $\cryptoPK{P}$
and a private (secret) key $\cryptoSK{P}$.
Using $\cryptoSK{P}$, a signature $\sigma := \cryptoSign{P}{x}$ for string $x$ can be created.

\begin{figure}%
    \centering
    \begin{minipage}[t]{\textwidth}
        Bitcoin Script implementation of output on $\mathsf{TX}_{\mathrm{settle}}$:
        \vspace{-0.5em}\ifthenelse{\equal{\jnPDFEXPORTVERSION}{springer}}{\scriptsize}{}
        \begin{Verbatim}[numbers=left,commandchars=\\\{\},codes={\catcode`$=3\catcode`^=7\catcode`_=8},frame=single]
IF
    ECEXP$\langle{}g\rangle{}$  PUSH$\langle{}H(p_i)\rangle{}$  EQUALVERIFY   \label{vrb:codetxfundshashchallenge}
    2  PUSH$\langle{}\mathsf{pk}_{P_{1,\mathrm{tmp}}}\rangle{}$  PUSH$\langle{}\mathsf{pk}_{P_{2,\mathrm{tmp}}}\rangle{}$  2  CHECKMULTISIGVERIFY   \label{vrb:codetxfundsredemption}
ELSE
    PUSH$\langle{}T_0{}+{}\Delta_{\mathrm{fwd}}\rangle{}$  CHECKLOCKTIMEVERIFY  DROP   \label{vrb:codetxfundstimelock}
    PUSH$\langle{}\mathsf{pk}_{P_1}\rangle{}$  CHECKSIGVERIFY   \label{vrb:codetxfundstimeout}
ENDIF
        \end{Verbatim}
    \end{minipage}
    \smallskip\\
    \begin{minipage}[t]{0.45\textwidth}
        Redemption for $P_1$ (`no' branch):
        \vspace{-0.5em}\ifthenelse{\equal{\jnPDFEXPORTVERSION}{springer}}{\scriptsize}{}
        \begin{Verbatim}[numbers=left,commandchars=\\\{\},codes={\catcode`$=3\catcode`^=7\catcode`_=8},frame=single]
$\mathsf{sig}_{P_1}(\mathsf{TX}_{\mathrm{payout},P_1})$
FALSE
        \end{Verbatim}
    \end{minipage}
    \hfill
    \begin{minipage}[t]{0.45\textwidth}
        Redemption for $P_2$ (`yes' branch):
        \vspace{-0.5em}\ifthenelse{\equal{\jnPDFEXPORTVERSION}{springer}}{\scriptsize}{}
        \begin{Verbatim}[numbers=left,commandchars=\\\{\},codes={\catcode`$=3\catcode`^=7\catcode`_=8},frame=single]
$\mathsf{sig}_{P_{1,\mathrm{tmp}}}(\mathsf{TX}_{\mathrm{retaliate}})$
$\mathsf{sig}_{P_{2,\mathrm{tmp}}}(\mathsf{TX}_{\mathrm{retaliate}})$
$\hat{p}_i$
TRUE
        \end{Verbatim}
    \end{minipage}
    \ifthenelse{\equal{\jnPDFEXPORTVERSION}{springer}}{\medskip}{\bigskip}\\
    \begin{minipage}[t]{\textwidth}
        Bitcoin Script implementation of output on $\mathsf{TX}_{\mathrm{retaliate}}$:
        \vspace{-0.5em}\ifthenelse{\equal{\jnPDFEXPORTVERSION}{springer}}{\scriptsize}{}
        \begin{Verbatim}[numbers=left,commandchars=\\\{\},codes={\catcode`$=3\catcode`^=7\catcode`_=8},frame=single]
IF
    ECEXP$\langle{}g\rangle{}$  PUSH$\langle{}H(\alpha_0)\rangle{}$  EQUALVERIFY   \label{vrb:codetxfundshashchallengeTWO}
    PUSH$\langle{}\mathsf{pk}_{P_1}\rangle{}$  CHECKSIGVERIFY   \label{vrb:codetxfundsredemptionTWO}
ELSE
    PUSH$\langle{}T_0{}+{}\Delta_{\mathrm{rev}}\rangle{}$  CHECKLOCKTIMEVERIFY  DROP   \label{vrb:codetxfundstimelockTWO}
    PUSH$\langle{}\mathsf{pk}_{P_2}\rangle{}$  CHECKSIGVERIFY   \label{vrb:codetxfundstimeoutTWO}
ENDIF
        \end{Verbatim}
    \end{minipage}
    \smallskip\\
    \begin{minipage}[t]{0.45\textwidth}
        Redemption for $P_1$ (`yes' branch):
        \vspace{-0.5em}\ifthenelse{\equal{\jnPDFEXPORTVERSION}{springer}}{\scriptsize}{}
        \begin{Verbatim}[numbers=left,commandchars=\\\{\},codes={\catcode`$=3\catcode`^=7\catcode`_=8},frame=single]
$\mathsf{sig}_{P_1}(\mathsf{TX}_{\mathrm{payout},P_1})$
$\hat{\alpha}_0$
TRUE
        \end{Verbatim}
    \end{minipage}
    \hfill
    \begin{minipage}[t]{0.45\textwidth}
        Redemption for $P_2$ (`no' branch):
        \vspace{-0.5em}\ifthenelse{\equal{\jnPDFEXPORTVERSION}{springer}}{\scriptsize}{}
        \begin{Verbatim}[numbers=left,commandchars=\\\{\},codes={\catcode`$=3\catcode`^=7\catcode`_=8},frame=single]
$\mathsf{sig}_{P_2}(\mathsf{TX}_{\mathrm{payout},P_2})$
FALSE
        \end{Verbatim}
    \end{minipage}
    \caption{
    Implementation of the \$1 outputs
    (\cf \figref{flow-chart-funds}), and witness stacks for redemption, %
    for the two staggered $\mathsf{TX}_{\mathrm{settle}}$ and $\mathsf{TX}_{\mathrm{retaliate}}$ (\cf \figref{integration-with-eltoo}).
    For $\mathsf{TX}_{\mathrm{settle}}$:
    l.~\ref{vrb:codetxfundshashchallenge} enforces revelation of $\hat{p}_i$ s.t. $H(\hat{p}_i) = H(p_i)$,
    l.~\ref{vrb:codetxfundsredemption} requires signatures of the temporary identities $P_{i,\mathrm{tmp}}$
    (created separately for each instance of the contract, and used to sign $\mathsf{TX}_{\mathrm{retaliate}}$ as part of the commitment to a forwarding),
    l.~\ref{vrb:codetxfundstimelock} enforces the timelock of $T_0 + \Delta_{\mathrm{fwd}}$,
    l.~\ref{vrb:codetxfundstimeout} requires a signature of $P_1$.
    For $\mathsf{TX}_{\mathrm{retaliate}}$:
    l.~\ref{vrb:codetxfundshashchallengeTWO} enforces revelation of $\hat{\alpha}_0$ s.t. $H(\hat{\alpha}_0) = H(\alpha_0)$,
    l.~\ref{vrb:codetxfundsredemptionTWO} requires a signature of $P_1$,
    l.~\ref{vrb:codetxfundstimelockTWO} enforces the timelock of $T_0 + \Delta_{\mathrm{rev}}$,
    l.~\ref{vrb:codetxfundstimeoutTWO} requires a signature of $P_2$.
    }
    \label{fig:code-funds-INLINE}
\end{figure}

The implementation of \figref{flow-chart-funds} is provided
in \figref{code-funds-INLINE}.
The first condition of the flow chart in \figref{flow-chart-funds},
which captures the forward component of the Boomerang contract,
is implemented as an output of $\mathsf{TX}_{\mathrm{settle}}$.
If this condition is met,
the distribution of the funds is decided by an additional transaction
$\mathsf{TX}_{\mathrm{retaliate}}$
which implements the reverse component of the Boomerang contract.
$\mathsf{TX}_{\mathrm{retaliate}}$ is agreed-upon and signed by both parties
using temporary one-time identities 
$P_{i,\mathrm{tmp}}$ %
as part of the deployment of a Boomerang contract.
The implementation of \figref{flow-chart-txfee} is provided
in \figref{code-txfee} of \secref{appendix-implementation-boomerang-contract-bitcoin-script-via-ecexp}.

\subsection{Implementation Using Adaptor Signatures}
\label{sec:implementation-adaptor-signatures-schnorr-signatures}

The dependency on a new opcode \verb|ECEXP|$\langle{}g\rangle{}$
can be lifted
by replacing the hash-lock
with adaptor signatures
\cite{poelstra2018scriptless}
based on Schnorr signatures
\cite{maxwell_simple_2018}
or ECDSA signatures
\cite{ecdsa_adaptor_sigs}.
A primer on adaptor signatures is given in \secref{appendix-background-adaptor-signatures}.
Rather than one party revealing a preimage $x$ for a challenge $h := H(x)$
when publishing a $\mathsf{TX}$ to claim funds or advance a contract,
$P_1$ and $P_2$ exchange an adaptor signature $\sigma'$ for $\mathsf{TX}$
which either party can turn into a proper signature $\sigma$ once they obtain $x$.
Once $\mathsf{TX}$ gets published together with $\sigma$,
the other party can derive $x$ from $\sigma$ and $\sigma'$.

\subsection{Deployment of Boomerang on an Eltoo Payment Channel}
\label{sec:deployment-boomerang-eltoo}

\begin{figure}%
    \centering
    \begin{tikzpicture}[
        ]
        \mathligsoff
        
        \ifthenelse{\equal{\jnPDFEXPORTVERSION}{springer}}{}{\footnotesize}

        \draw [fill=black!10,draw=none] (-2,1.7) -- ++(6,0) -- ++(0,-3.5) -- ++(-1.7,0) -- ++(0,-1.2) -- ++(2,0) -- ++(0,-1.5) -- (-2,-4.5) coordinate (eltooareasouthwest) -- cycle;
        \draw [fill=red!15,draw=none] (-2,-4.6) coordinate (boomerangnorthwest) -- ++(6.4,0) -- ++(0,1.7) -- ++(-2,0) -- ++(0,1) -- ++(1.7,0) -- ++(0,3.6) -- ++(5.8,0) -- ++(0,-8.3) -- (-2,-6.6) -- cycle;
        \node [anchor=south west] at (eltooareasouthwest) {\emph{\large Eltoo}};
        \node [anchor=north west] at (boomerangnorthwest) {\emph{\large Boomerang}};

        \node (dotdotdot) at (-1.5,0) {...};

        \node (TXupdate) at (0,0) [draw,rectangle,minimum height=5em,minimum width=5em] {};
        \node at (TXupdate.north west) [anchor=south west] {$\mathsf{TX}_{\mathrm{update}}^{\mathrm{el2}}$};
        \draw [-Latex] (dotdotdot) -- (TXupdate) node [pos=1,below,anchor=north east,align=center] {\emojisig{}$P_{1,\mathrm{u}}^{\mathrm{el2}}$\\\emojisig{}$P_{2,\mathrm{u}}^{\mathrm{el2}}$};

        \node (TXsettle) at (2.3,-3) [draw,rectangle,minimum height=7em,minimum width=5em] {};
        \node at (TXsettle.north west) [anchor=south west] {$\mathsf{TX}_{\mathrm{settle}}^{\mathrm{el2}}$};

        \draw [-Latex,rounded corners] (TXupdate) -- ++(2,0) node [anchor=west] {\emojisigcheck{}$P_{1,\mathrm{u}}^{\mathrm{el2}},P_{2,\mathrm{u}}^{\mathrm{el2}}$};
        \draw [-Latex,rounded corners] (TXupdate) -| ++(1.10,-2) node [anchor=west,pos=0.63] {\emojireltl{} \emojisigcheck{}$P_{1,\mathrm{s}}^{\mathrm{el2}},P_{2,\mathrm{s}}^{\mathrm{el2}}$} |- (TXsettle.west) node [pos=1,below,anchor=north east,align=center] {\emojisig{}$P_{1,\mathrm{s}}^{\mathrm{el2}}$\\\emojisig{}$P_{2,\mathrm{s}}^{\mathrm{el2}}$};

        \draw [-Latex] (TXsettle.south east) ++(0,0.8) coordinate (TXsettleOUT3) -- ++(0.3,0) node [anchor=west] {\emojisigcheck $P_1$};
        \draw [-Latex] (TXsettle.south east) ++(0,0.3) coordinate (TXsettleOUT4) -- ++(0.3,0) node [anchor=west] {\emojisigcheck $P_2$};

        \draw (TXsettle.north east) ++(0,-0.3) coordinate (TXsettleOUT1) node [anchor=east] {$\delta$};
        \draw [-Latex] (TXsettleOUT1) -- ++(3,0) node [anchor=west] {\emojihash{}$H(p_i)$ \emojisigcheck{}$P_2$};
        \draw [-Latex,rounded corners] (TXsettleOUT1) -| ++(1.5,0.5) -- ++(1.5,0) node [anchor=west] {\emojiabstl{}\,$T_0{+}\Delta_{\mathrm{fwd}}$ \emojisigcheck{}$P_1$};

        \node (TXretaliate) at (5.8,-5) [draw,rectangle,minimum height=5em,minimum width=5em] {};
        \node at (TXretaliate.north west) [anchor=south west] {$\mathsf{TX}_{\mathrm{retaliate}}^{\mathrm{boom}}$};

        \draw [-Latex,rounded corners] (TXsettle.north east) ++(0,-0.8) coordinate (TXsettleOUT2) node [anchor=east] {\$1} -| ++(1.5,-1) node [anchor=west,pos=0.75] {\emojihash{}$H(p_i)$ \emojisigcheck{}$P_{1,\mathrm{tmp}}^{\mathrm{boom}},P_{2,\mathrm{tmp}}^{\mathrm{boom}}$} |- (TXretaliate) node [pos=1,below,anchor=north east,align=right] {\emojipreimage{}$P(i)$\\\emojisig{}$P_{1,\mathrm{tmp}}^{\mathrm{boom}}$\\\emojisig{}$P_{2,\mathrm{tmp}}^{\mathrm{boom}}$};
        \draw [-Latex] (TXsettleOUT2) -- ++(3,0) node [anchor=west] {\emojiabstl{}\,$T_0{+}\Delta_{\mathrm{fwd}}$ \emojisigcheck{}$P_1$};

        \draw [-Latex,rounded corners] (TXretaliate.east) -| ++(0.25,0.25) -- ++(0.35,0) node [anchor=west] {\emojihash{}$H(\alpha_0)$ \emojisigcheck{}$P_1$};
        \draw [-Latex,rounded corners] (TXretaliate.east) -| ++(0.25,-0.25) -- ++(0.35,0) node [anchor=west] {\emojiabstl{}\,$T_0{+}\Delta_{\mathrm{rev}}$ \emojisigcheck{}$P_2$};
        
    \end{tikzpicture}
    \caption[]{
        A Boomerang contract is deployed as outputs and a thereon depending $\mathsf{TX}$
        on top of the settlement mechanism of an Eltoo PC.
        (Legend:
        Boxes are Bitcoin transactions,
        an arrow leaving/entering a box indicates an output/input, respectively,
        labels along arrows indicate spending conditions,
        forks of arrows indicate alternative spending conditions,
        labels near $\mathsf{TX}$ inputs indicate redemption witnesses,
        \emojisigcheck{} `requires signature of',
        \emojisig{} `signed by',
        \emojihash{} `preimage challenge',
        \emojipreimage{} `preimage solution',
        \emojiabstl{} `absolute time-lock',
        \emojireltl{} `relative time-lock'.)
    }
    \label{fig:integration-with-eltoo}
\end{figure}
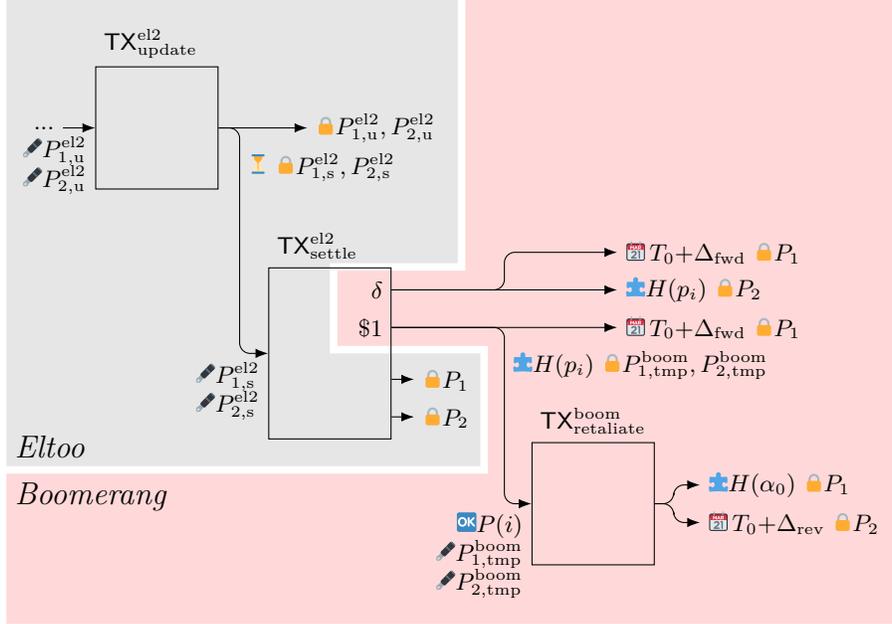

\Figref{integration-with-eltoo}
shows a Boomerang contract for the $i$-th TX with preimage challenge $H(p_i)$ deployed on top of an Eltoo PC.
Transactions and cryptographic identities belonging to Eltoo are marked with $(.)^\mathrm{el2}$,
those belonging to Boomerang are marked with $(.)^\mathrm{boom}$.
\Figref{integration-with-eltoo} also illustrates \figref{code-txfee,code-funds-INLINE}.

\section{Experimental Evaluation}
\label{sec:experimental-evaluation}

We evaluate Boomerang in a low-fidelity prototype implementation%
\footnote{
The source code is available on: \url{https://github.com/tse-group/boomerang}
}
of the routing components
inspired by the testbed of Flash \cite{wang_flash:_2019}.
First, we outline our experimental setup,
then we specify the three contending schemes,
and finally we present and discuss the observed performance.

\subsection{Experimental Setup}
\label{sec:scenario}

The PN topology is drawn %
from the Watts–Strogatz ensemble
where initially $N=100$ nodes %
in a regular ring lattice
are connected to their $8$ nearest neighbors,
and subsequently each initial edge
is rewired
randomly
with probability $0.8$.
For each PC endpoint,
its initial liquidity is drawn
log-uniformly %
in $[\log(100), \log(1000)]$.
The topology is static and known to all nodes;
each PC's balance %
is only
known to the PC's endpoints.
$50000$ TFs are generated as follows.
Source and destination are sampled uniformly %
from the $N$ nodes.
The amounts are drawn from the
Ripple dataset used in SpeedyMurmurs \cite{roos_settling_2017}.
As in previous works,
each node has a backlog of TFs it attempts to route one by one.
We report
sample mean and sample standard deviation
of
$10$ PNs and TF traces.

Our implementation is a low-fidelity prototype of the routing and communications tasks.
An abstract PN protocol %
accommodates
different routing schemes
(\cf Flash \cite{wang_flash:_2019}).
A TX takes place in two phases:
First, in the \textsc{Reserve} phase,
Boomerang contracts are set up along a path from $A$ to $B$ if PC liquidity permits.
$A$ is notified whether \textsc{Reserve} was successful.
Second, in the \textsc{Rollback}/\textsc{Execute} phase,
the chain of Boomerang contracts is either dismantled (\textsc{Rollback})
or the funds are delivered (\textsc{Execute}).
An \textsc{Abort} message can %
stop an ongoing \textsc{Reserve} attempt.
Purely informational messages
(\eg, outcome of \textsc{Reserve})
are relayed without delay.
Operations on the PC
(\eg, deploying a Boomerang contract)
are simulated by a uniform %
delay from \SIrange{50}{150}{\milli\second}.
Note that PC balances are only discovered implicitely
through failed/successful TX attempts,
similar to the current state of affairs in Lightning.

The software is written in Golang 1.12.7.
Every node is an independent process on a machine with
4x AMD Opteron 6378 processors (total 64 cores).
The scenario is chosen such that the CPUs are never fully utilized
to avoid distortion from computational limitations.
The nodes communicate via `localhost'.

\subsection{Three Simple Routing Protocols}
\label{sec:three-routing-protocols}

We compare three multi-path routing protocols
(pseudo code given in %
\secref{appendix-pseudo-code-evaluated-routing-schemes}).
`Retry' (\cf \algoref{retry}) initially attempts $v$ TXs and reattempts up to $u$ of them.
`Redundancy' (\cf \algoref{redundancy}) attempts $v+u$ TXs from the start.
`Redundant-Retry(10)' (\cf \algoref{redundant-retry-10}) is a combination of the two, which starts out with $v+\min(u,10)$ TXs and reattempts up to $u-\min(u,10)$ of them.
The aim is to trade off the lower TTC of `Redundancy'
with the adaptivity of `Retry'.
TXs are routed on paths chosen randomly from a set of precomputed edge-disjoint shortest paths.
All three schemes
use atomic multi-path (AMP, \cite{osuntokun_[lightning-dev]_2018}),
\ie, they only \textsc{Execute} once enough successful TX attempts have been made to
satisfy the full TF (\cf \algoref{amp-utilities}).
If the TF cannot be fully satisfied, all TXs are \textsc{Rollback}ed.

The baseline protocol `Retry' is rather simplistic
compared to recent developments
\cite{%
poon_bitcoin_2016-1,%
prihodko_flare:_2016,%
malavolta_silentwhispers:_2016,%
roos_settling_2017,%
hoenisch_aodvbased_2018,%
sivaraman_routing_2018,%
di_stasi_routing_2018,%
wang_flash:_2019%
}.
We choose it because it resembles
what currently can and is being done on Lightning.
We conjecture that redundancy as a generic technique
can boost AMP routing algorithms across the board.

\subsection{Simulation Results}
\label{sec:simulation-results}

\begin{figure}[tb!]
    \begin{center}
        \ref{leg:sim-results-small-ttc-retry-02-amp-2} Retry
        \quad
        \ref{leg:sim-results-small-ttc-redundancy-02-amp-2} Redundancy
        \quad
        \ref{leg:sim-results-small-ttc-redundantretry-02-amp-2-10} Redundant-Retry(10)
    \end{center}
    \vspace{-0.5em}
    \begin{subfigure}[t]{0.32\linewidth}
        \centering
        \begin{tikzpicture}
            \begin{axis}[
                    myresultsplot01,
                    xlabel={$u$},
                    symbolic x coords={0,10,20,75,150},
                    enlarge x limits=0.12,
                    bar width=1.1mm,
                    width=1.175\linewidth,   %
                    height=1.1\linewidth,
                ]
    
                \def\PLOTDATAFILE{figures/02_nodes100_txs500_paths25_edges4.605170to6.907755_paths25_SELECT.data}
    
                \foreach \thisColor/\thisLabel/\thisAlgorithm in {
                    myParula01Blue/Retry/retry-02-amp-2,
                    myParula02Orange/Redundancy/redundancy-02-amp-2,
                    myParula03Yellow/Redundant-Retry(10)/redundantretry-02-amp-2-10} {
                        \edef\temp{
                            \noexpand\addplot+ [
                                error bars/.cd,
                                y explicit,
                                y dir=both,
                            ] table [
                                x=u,
                                y expr=\noexpand\thisrow{throughput_success-mean},
                                y error expr=\noexpand\thisrow{throughput_success-std},
                                discard if not={algo}{\thisAlgorithm},
                            ] {\PLOTDATAFILE};
                            \noexpand\addlegendentry{\thisLabel};
                        }
                        \temp
                }
                
                \legend{};
    
            \end{axis}
        \end{tikzpicture}
        \caption{
        Throughput [\si{\xrp\per\second\per\nde}]
        }
        \label{fig:sim-results-small-throughput}
    \end{subfigure}
    \begin{subfigure}[t]{0.32\linewidth}
        \centering
        \begin{tikzpicture}
            \begin{axis}[
                    myresultsplot01,
                    xlabel={$u$},
                    symbolic x coords={0,10,20,75,150},
                    enlarge x limits=0.12,
                    bar width=1.1mm,
                    width=1.175\linewidth,   %
                    height=1.1\linewidth,
                ]
    
                \def\PLOTDATAFILE{figures/02_nodes100_txs500_paths25_edges4.605170to6.907755_paths25_SELECT.data}
    
                \foreach \thisColor/\thisLabel/\thisAlgorithm in {
                    myParula01Blue/Retry/retry-02-amp-2,
                    myParula02Orange/Redundancy/redundancy-02-amp-2,
                    myParula03Yellow/Redundant-Retry(10)/redundantretry-02-amp-2-10} {
                        \edef\temp{
                            \noexpand\addplot+ [
                                error bars/.cd,
                                y explicit,
                                y dir=both,
                            ] table [
                                x=u,
                                y expr=\noexpand\thisrow{ttc_for_successful_tx-mean},
                                y error expr=\noexpand\thisrow{ttc_for_successful_tx-std},
                                discard if not={algo}{\thisAlgorithm},
                            ] {\PLOTDATAFILE};
                            \noexpand\addlegendentry{\thisLabel};
                            \noexpand\label{leg:sim-results-small-ttc-\thisAlgorithm};
                        }
                        \temp
                }
                
                \legend{};
    
            \end{axis}
        \end{tikzpicture}
        \caption{
        TTC [\si{\second}]
        }
        \label{fig:sim-results-small-ttc}
    \end{subfigure}
    \begin{subfigure}[t]{0.32\linewidth}
        \centering
        \begin{tikzpicture}
            \begin{axis}[
                    myresultsplot01,
                    xlabel={$u$},
                    symbolic x coords={0,10,20,75,150},
                    enlarge x limits=0.12,
                    bar width=1.1mm,
                    width=1.175\linewidth,   %
                    height=1.1\linewidth,
                ]
    
                \def\PLOTDATAFILE{figures/02_nodes100_txs500_paths25_edges4.605170to6.907755_paths25_SELECT.data}
    
                \foreach \thisColor/\thisLabel/\thisAlgorithm in {
                    myParula01Blue/Retry/retry-02-amp-2,
                    myParula02Orange/Redundancy/redundancy-02-amp-2,
                    myParula03Yellow/Redundant-Retry(10)/redundantretry-02-amp-2-10} {
                        \edef\temp{
                            \noexpand\addplot+ [
                                error bars/.cd,
                                y explicit,
                                y dir=both,
                            ] table [
                                x=u,
                                y expr=\noexpand\thisrow{volume_for_successful_tx-mean},
                                y error expr=\noexpand\thisrow{volume_for_successful_tx-std},
                                discard if not={algo}{\thisAlgorithm},
                            ] {\PLOTDATAFILE};
                            \noexpand\addlegendentry{\thisLabel};
                        }
                        \temp
                }
                
                \legend{};
    
            \end{axis}
        \end{tikzpicture}
        \caption{
        Volume
        [\si{\xrp}]
        }
        \label{fig:sim-results-small-size}
    \end{subfigure}
    \caption{
    \subref{fig:sim-results-small-throughput}
    Average success throughput per node,
    \subref{fig:sim-results-small-ttc}
    Average time-to-completion for successful TFs,
    \subref{fig:sim-results-small-size}
    Average volume for successful TFs.
    }
    \label{fig:sim-results-small}
\end{figure}

Our results are shown in
\figref{sim-results-small}.
Supplemental plots can be found in
\ifthenelse{\equal{\jnPDFEXPORTVERSION}{springer}}{\cite[Section E, Fig.~12]{boomerang_arxiv_full}.}{\secref{appendix-plots-experimental-evaluation}.}
Previous works have gauged the performance of PN routing algorithms
in terms of success volume or success count, \ie,
for a given trace of TFs,
what total amount or number of TFs
gets satisfied. %
However, these metrics are problematic. %
If TFs are spread out in time
then more liquidity is available in the PN.
This renders it easier to satisfy a TF
and inflates these metrics.

Instead,
we consider success throughput, \ie,
total amount of successfully transferred funds
\emph{per runtime}.
The results for $v=25$ are shown
in \figref{sim-results-small-throughput}
\ifthenelse{\equal{\jnPDFEXPORTVERSION}{springer}}{(\cf \cite[Fig.~12(a)]{boomerang_arxiv_full}).}{(\cf \figref{sim-results-throughput}).}
`Redundancy' and `Redundant-Retry(10)'
show a 2x increase in throughput
compared to `Retry'.

Another relevant figure of merit is
the average time-to-completion (TTC) of a successful TF, \ie,
how long is the delay between when the execution of a TF starts
and when AMP is finalized.
This determines the latency experienced by the user
and for how long liquidity is tied up in pending AMP TFs.
The results for $v=25$
in \figref{sim-results-small-ttc}
\ifthenelse{\equal{\jnPDFEXPORTVERSION}{springer}}{(\cf \cite[Fig.~12(b)]{boomerang_arxiv_full})}{(\cf \figref{sim-results-ttc})}
show a \SI{40}{\percent} reduction in TTC
for `Redundancy' over `Retry'.
The larger $u$, the longer `Retry' takes but the quicker `Redundancy' completes.
This plot also demonstrates how `Redundant-Retry(10)'
trades off between `Retry' and `Redundancy'.
For $u<=10$,
`Redundant-Retry(10)' is identical with `Redundancy'
and hence follows the performance improvement of `Redundancy'.
For $u>10$,
the retry aspect weighs in more and more and
`Redundant-Retry(10)' follows a similar trajectory as `Retry'.

Finally,
\figref{sim-results-small-size}
\ifthenelse{\equal{\jnPDFEXPORTVERSION}{springer}}{(\cf \cite[Fig.~12(c)]{boomerang_arxiv_full})}{(\cf \figref{sim-results-size})}
shows the average size of a successful TF,
where `Redundant-Retry(10)' outperforms `Redundancy' by \SI{30}{\percent}
which in turn outperforms `Retry' by \SI{40}{\percent}.
Note that the throughput of `Redundant-Retry(10)' and `Redundancy' are comparable.
Thus,
weighting between
adaptively retrying
and upfront redundancy
trades off
TTC
and
average successful TF volume,
at a constant throughput.

\section*{Acknowledgments}

{
\ifthenelse{\equal{\jnPDFEXPORTVERSION}{springer}}{\footnotesize}{}
We thank Giulia Fanti and Lei Yang for fruitful discussions.
VB and DT
are supported by the Center for Science of Information,
an NSF Science and Technology Center,
under grant agreement CCF-0939370.
JN
is supported by
the Reed-Hodgson Stanford Graduate Fellowship.
Icons from `Twemoji v12.0' (\url{https://github.com/twitter/twemoji}) by Twitter, Inc and other contributors, licensed under CC BY 4.0.
\par
}

\ifthenelse{\equal{\jnPDFEXPORTVERSION}{springer}}{\vspace{-0.5em}}{}

\bibliographystyle{splncs04}
\bibliography{references_zotero,references_manual,references_ln_issues}
\appendix

\section{Cryptographic Preliminaries}
\label{sec:appendix-cryptographic-preliminaries}

Let $\IG$ be a cyclic multiplicative group of prime order $q \geq 2^{2\lambda}$ with a generator $g \in \IG$, where $\lambda$ is a security parameter.
Let $H\colon \IZ_q \to \IG$ with $H(x) := g^x$, where $\IZ_q$ is the finite field of size $q$ (\ie, integers modulo $q$).
We require that $H$ be difficult to invert,
which is formalized in the following two definitions:
\begin{definition}[Negligible Function]
    A function $\varepsilon\colon \IN \to \IR^{+}$ is %
    \emph{negligible}
    if
    \begin{IEEEeqnarray}{C}
        \forall c > 0\colon
        \exists k_0\colon
        \forall k > k_0\colon
        \qquad
        \varepsilon(k) < \frac{1}{k^c}.
    \end{IEEEeqnarray}
\end{definition}
In other words,
negligible is what decays faster than every polynomial.
\begin{definition}[Discrete Logarithm (DL) Assumption]
    Given a generator $g$ of a group $\IG$,
    and an $x \drawrandom \IZ_q$
    chosen uniformly at random in $\IZ_q$,
    for every probabilistic polynomial time (with respect to $\lambda$) algorithm $\mathcal A_{\mathrm{DL}}$,
    \begin{IEEEeqnarray}{C}
        \operatorname{Pr}\left[ \mathcal A_{\mathrm{DL}}(g, g^x) = x \right] = \varepsilon(\lambda).
    \end{IEEEeqnarray}
\end{definition}
The discrete logarithm problem (DLP) is said to be hard for generator $g$ in group $\IG$,
if the DL assumption holds for $g$ and $\IG$,
\ie,
no computationally bounded adversary can compute $\log_g(g^x)$
with non-negligible probability.
It is commonly assumed that the DLP is hard in certain elliptic curves (ECs),
which are hence widely used in cryptographic applications, \eg, in Bitcoin.
The DL assumption makes $H$ a one-way function.
\section{Implementation of Boomerang Contract in Bitcoin Script via Elliptic Curve Scalar Multiplication} %
\label{sec:appendix-implementation-boomerang-contract-bitcoin-script-via-ecexp}

See \figref{code-funds-INLINE,code-txfee} for Bitcoin Script implementations of \figref{flow-charts}.

\begin{figure}
    \centering
    \begin{minipage}[t]{\textwidth}
        Bitcoin Script implementation of output on $\mathsf{TX}_{\mathrm{settle}}$:
        \vspace{-0.5em}\ifthenelse{\equal{\jnPDFEXPORTVERSION}{springer}}{\scriptsize}{}
        \begin{Verbatim}[numbers=left,commandchars=\\\{\},codes={\catcode`$=3\catcode`^=7\catcode`_=8},frame=single]
IF
    ECEXP$\langle{}g\rangle{}$  PUSH$\langle{}H(p_i)\rangle{}$  EQUALVERIFY   \label{vrb:codetxfeehashchallenge}
    PUSH$\langle{}\mathsf{pk}_{P_2}\rangle{}$  CHECKSIGVERIFY   \label{vrb:codetxfeeredemptionPtwo}
ELSE
    PUSH$\langle{}T_0{}+{}\Delta_{\mathrm{fwd}}\rangle{}$  CHECKLOCKTIMEVERIFY  DROP   \label{vrb:codetxfeetimelock}
    PUSH$\langle{}\mathsf{pk}_{P_1}\rangle{}$  CHECKSIGVERIFY   \label{vrb:codetxfeeredemptionPone}
ENDIF
        \end{Verbatim}
    \end{minipage}
    \smallskip\\
    \begin{minipage}[t]{0.45\textwidth}
        Redemption for $P_1$ (`no' branch):
        \vspace{-0.5em}\ifthenelse{\equal{\jnPDFEXPORTVERSION}{springer}}{\scriptsize}{}
        \begin{Verbatim}[numbers=left,commandchars=\\\{\},codes={\catcode`$=3\catcode`^=7\catcode`_=8},frame=single]
$\mathsf{sig}_{P_1}(\mathsf{TX}_{\mathrm{payout},P_1})$
FALSE
        \end{Verbatim}
    \end{minipage}
    \hfill%
    \begin{minipage}[t]{0.45\textwidth}
        Redemption for $P_2$ (`yes' branch):
        \vspace{-0.5em}\ifthenelse{\equal{\jnPDFEXPORTVERSION}{springer}}{\scriptsize}{}
        \begin{Verbatim}[numbers=left,commandchars=\\\{\},codes={\catcode`$=3\catcode`^=7\catcode`_=8},frame=single]
$\mathsf{sig}_{P_2}(\mathsf{TX}_{\mathrm{payout},P_2})$
$\hat{p}_i$
TRUE
        \end{Verbatim}
    \end{minipage}
    \caption{Bitcoin Script implementation of the output concerning $\delta$ TX fee (\cf \figref{flow-chart-txfee}), and witness stacks for redemption in favor of $P_1$ and $P_2$, respectively:
    l.~\ref{vrb:codetxfeehashchallenge} enforces revelation of $\hat{p}_i$ such that $H(\hat{p}_i) = H(p_i)$,
    l.~\ref{vrb:codetxfeeredemptionPtwo} requires a signature of $P_2$,
    l.~\ref{vrb:codetxfeetimelock} enforces the timelock until $T_0 + \Delta_{\mathrm{fwd}}$,
    l.~\ref{vrb:codetxfeeredemptionPone} requires a signature of $P_1$.
    }
    \label{fig:code-txfee}
\end{figure}

\section{Background on Adaptor Signatures}
\label{sec:appendix-background-adaptor-signatures}

We briefly summarize Schnorr signatures
\cite{maxwell_simple_2018}.
Let $\tilde{H}$ be a cryptographic hash function
(modeled as a random oracle),
and $x \| y$ denote the concatenation of $x$ and $y$.
We continue to assume that $\IG$ is a multiplicative
group with group operation `$\cdot$'.
For Schnorr signatures,
every identity is composed of a secret key $x$ and a public key $P := g^x$.
To sign a message $m$,
draw $r \drawrandom \IZ_q$,
then compute $R := g^r$ and
$s = r + \tilde{H}(P \| R \| m) x$.
The signature is $\sigma := (s, R)$.
To verify a signature $\sigma := (s, R)$
for $m$ by $P$,
check
\begin{IEEEeqnarray}{C}
    g^s \overset{?}{=} R \cdot P^{\tilde{H}(P \| R \| m)}.
\end{IEEEeqnarray}

An adaptor signature $\sigma'$ has the property that given $\sigma'$,
knowledge of a proper signature $\sigma$ is equivalent to knowledge of a precommitted value $t$
\cite{poelstra2018scriptless}.
Consider parties $P_1$ and $P_2$
with secret keys $x_i$ and public keys $P_i := g^{x_i}$.
Both know a commitment $T := g^t$ to a (potentially unknown) value $t$.
To create an adaptor signature $\sigma'$
for $m$,
both draw $r_i \drawrandom \IZ_q$, compute $R_i := g^{r_i}$,
and exchange $(P_i, R_i)$.
Then, they compute and exchange
\begin{IEEEeqnarray}{C}
    s_i' = r_i + \tilde{H}(P_1 \cdot P_2 \| R_1 \cdot R_2 \cdot T \| m) x_i.
\end{IEEEeqnarray}
The adaptor signature is $\sigma' = (R_1 \cdot R_2 \cdot T, s_1' + s_2')$.
If either $P_i$ gets to know $t$,
they can produce a valid total signature $\sigma = (R_1 \cdot R_2 \cdot T, s_1' + s_2' + t)$.
Vice versa,
if either $P_i$ learns a valid total signature $\sigma = (R_1 \cdot R_2 \cdot T, s)$,
they can compute $t = s - s_1' - s_2'$.
For instance,
suppose
$m$ is a transaction that benefits $P_2$ and requires a signature from $P_1 \cdot P_2$
with nonce $R_1 \cdot R_2 \cdot T$.
Furthermore,
suppose $P_2$ obtains $t$.
Then it can use the adaptor signature $\sigma'$
to produce a valid total signature $\sigma$ and claim its funds.
In this case, $P_1$ can recover $t$ from $\sigma$ and $\sigma'$.

\section{Pseudo Code of Evaluated Routing Schemes}
\label{sec:appendix-pseudo-code-evaluated-routing-schemes}

See \algoref{amp-utilities,retry,redundancy,redundant-retry-10}.

\begin{algorithm}
    \caption{Finalize AMP TF}
    \label{algo:amp-utilities}
    \ifthenelse{\equal{\jnPDFEXPORTVERSION}{springer}}{\scriptsize}{}
    \begin{algorithmic}[1]
        \Procedure{FinalizeAMP}{$\CO, \CP, v$}
                        \Comment{\textsc{Abort}s outstanding TX attempts $\CO$, \textsc{Execute}s/\textsc{Rollback}s successful TXs $\CP$ depending on number of required TXs $v$}
            \State $\operatorname{Send\textsc{Abort}}(\CO)$
            \If{$\cardinality{\CP} = v$}
                \State $\operatorname{Send\textsc{Execute}}(\CP)$
            \Else
                \State $\operatorname{Send\textsc{Rollback}}(\CP)$
            \EndIf
            \State $\operatorname{Receive\textsc{Reserve}ResponsesAndSend\textsc{Rollback}}(\CO)$
        \EndProcedure
    \end{algorithmic}
\end{algorithm}

\begin{algorithm}
    \caption{`Retry' routing scheme}
    \label{algo:retry}
    \ifthenelse{\equal{\jnPDFEXPORTVERSION}{springer}}{\scriptsize}{}
    \begin{algorithmic}[1]
        \Procedure{Retry}{$\CT, v, u$}
                        \Comment{Split each TF $t\in\CT$ into $v$ TXs and retry $u$ TXs}
            \For{$t\in\CT$}
                            \Comment{Source $t.s$, destination $t.d$, amount $t.v$}
                \State $\CO, \CP, \CN \gets \emptyset, \emptyset, \emptyset$
                            \Comment{TX attempts: outstanding $\CO$, successful $\CP$, failed $\CN$}
                \State $u' \gets u$
                            \Comment{Unused TX retries}
                \For{$i=1,...,v$}
                                \Comment{Attempt $v$ TXs}
                    \State $\CO \gets \CO \cup \operatorname{Send\textsc{Reserve}}(\operatorname{RandomPath}(t.s, t.d), t.v/v)$
                \EndFor
                \While{
                    $\cardinality{\CP} < v$
                    $\land$
                    $\cardinality{\CO} > 0$
                    $\land$
                    $\cardinality{\CP} + \cardinality{\CO} + u' >= v$
                    }
                                \Comment{TF is contingent}
                        \State $\operatorname{ReceiveAndClassify\textsc{Reserve}Responses}(\CO, \CP, \CN)$
                        \For{new elements $r\in\CN$}
                            \State $\operatorname{Send\textsc{Rollback}}(r)$
                            \If{$u' > 0$}
                                \State $\CO \gets \CO \cup \operatorname{Send\textsc{Reserve}}(\operatorname{RandomPath}(t.s, t.d), t.v/v)$
                                \State $u' \gets u' - 1$
                            \EndIf
                        \EndFor
                \EndWhile
                \State $\operatorname{\textsc{FinalizeAMP}}(\CO, \CP, v)$
            \EndFor
        \EndProcedure
    \end{algorithmic}
\end{algorithm}

\begin{algorithm}
    \caption{`Redundancy' routing scheme}
    \label{algo:redundancy}
    \ifthenelse{\equal{\jnPDFEXPORTVERSION}{springer}}{\scriptsize}{}
    \begin{algorithmic}[1]
        \Procedure{Redundancy}{$\CT, v, u$}
                        \Comment{Split TFs $t\in\CT$ into $v$ plus $u$ redundant TXs}
            \For{$t\in\CT$}
                            \Comment{Source $t.s$, destination $t.d$, amount $t.v$}
                \State $\CO, \CP, \CN \gets \emptyset, \emptyset, \emptyset$
                            \Comment{TX attempts: outstanding $\CO$, successful $\CP$, failed $\CN$}
                \For{$i=1,...,v+u$}
                                \Comment{Attempt $v$ TXs}
                    \State $\CO \gets \CO \cup \operatorname{Send\textsc{Reserve}}(\operatorname{RandomPath}(t.s, t.d), t.v/v)$
                \EndFor
                \While{
                    $\cardinality{\CP} < v$
                    $\land$
                    $\cardinality{\CO} > 0$
                    $\land$
                    $\cardinality{\CP} + \cardinality{\CO} >= v$
                    }
                                \Comment{TF is contingent}
                        \State $\operatorname{ReceiveAndClassify\textsc{Reserve}Responses}(\CO, \CP, \CN)$
                        \For{new elements $r\in\CN$}
                            \State $\operatorname{Send\textsc{Rollback}}(r)$
                        \EndFor
                \EndWhile
                \State $\operatorname{\textsc{FinalizeAMP}}(\CO, \CP, v)$
            \EndFor
        \EndProcedure
    \end{algorithmic}
\end{algorithm}

\begin{algorithm}
    \caption{`Redundant-Retry(10)' routing scheme}
    \label{algo:redundant-retry-10}
    \ifthenelse{\equal{\jnPDFEXPORTVERSION}{springer}}{\scriptsize}{}
    \begin{algorithmic}[1]
        \Procedure{RedundantRetry}{$\CT, v, u, 10$}
            \For{$t\in\CT$}
                            \Comment{Source $t.s$, destination $t.d$, amount $t.v$}
                \State $\CO, \CP, \CN \gets \emptyset, \emptyset, \emptyset$
                            \Comment{TX attempts: outstanding $\CO$, successful $\CP$, failed $\CN$}
                \State $u' \gets u - \operatorname{min}(u, 10)$
                            \Comment{Unused TX retries}
                \For{$i=1,...,v+\operatorname{min}(u, 10)$}
                                \Comment{Attempt $v+\operatorname{min}(u, 10)$ TXs}
                    \State $\CO \gets \CO \cup \operatorname{Send\textsc{Reserve}}(\operatorname{RandomPath}(t.s, t.d), t.v/v)$
                \EndFor
                \While{
                    $\cardinality{\CP} < v$
                    $\land$
                    $\cardinality{\CO} > 0$
                    $\land$
                    $\cardinality{\CP} + \cardinality{\CO} + u' >= v$
                    }
                                \Comment{TF is contingent}
                        \State $\operatorname{ReceiveAndClassify\textsc{Reserve}Responses}(\CO, \CP, \CN)$
                        \For{new elements $r\in\CN$}
                            \State $\operatorname{Send\textsc{Rollback}}(r)$
                            \If{$u' > 0$}
                                \State $\CO \gets \CO \cup \operatorname{Send\textsc{Reserve}}(\operatorname{RandomPath}(t.s, t.d), t.v/v)$
                                \State $u' \gets u' - 1$
                            \EndIf
                        \EndFor
                \EndWhile
                \State $\operatorname{\textsc{FinalizeAMP}}(\CO, \CP, v)$
            \EndFor
        \EndProcedure
    \end{algorithmic}
\end{algorithm}

\ifthenelse{\equal{\jnPDFEXPORTVERSION}{springer}}{}{

\section{Supplemental Plots for Experimental Evaluation}
\label{sec:appendix-plots-experimental-evaluation}

For a more detailed version of \figref{sim-results-small},
see \figref{sim-results}.

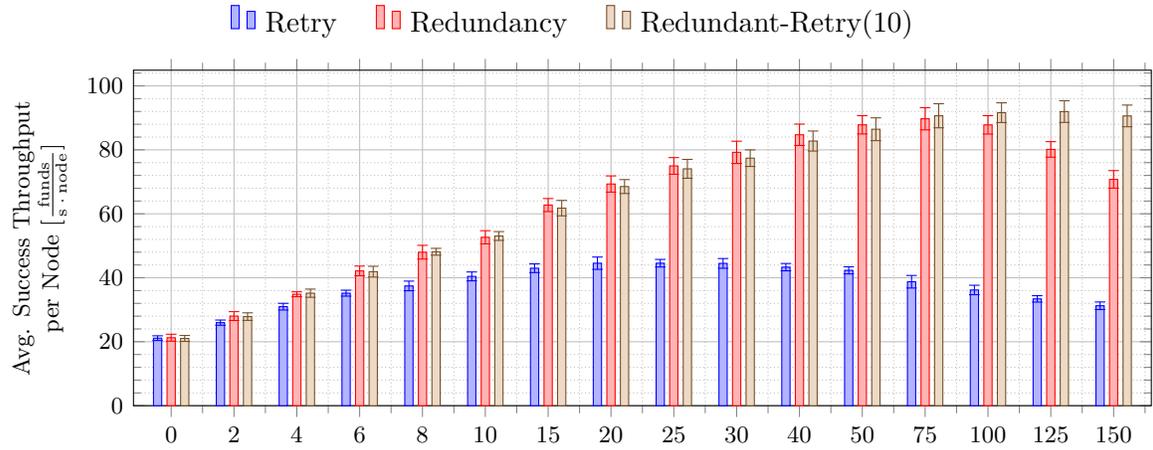
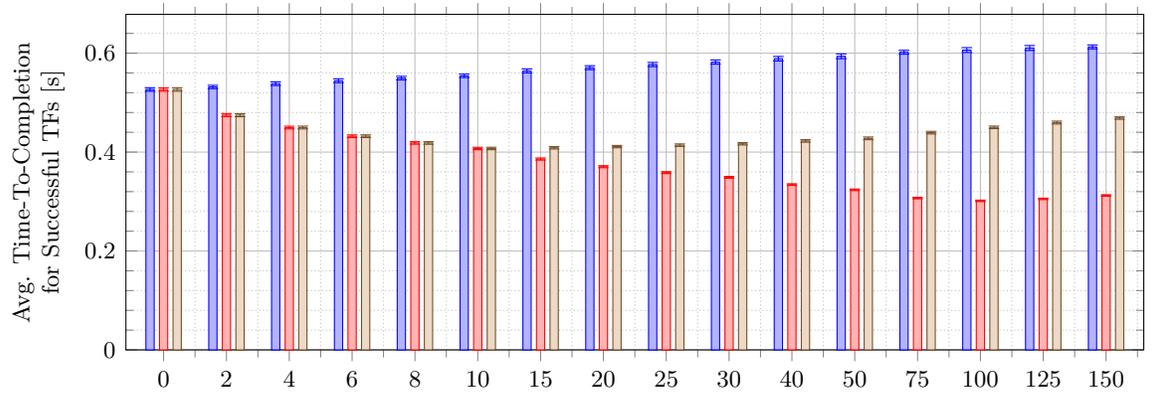
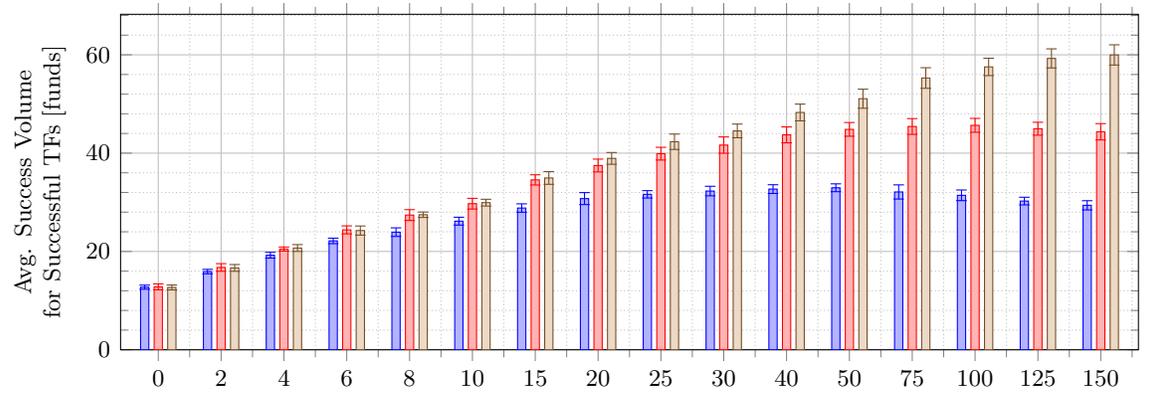
\begin{figure}[p!]%
    \begin{center}
        \ref{leg:sim-results-small-ttc-retry-02-amp-2} Retry
        \quad
        \ref{leg:sim-results-small-ttc-redundancy-02-amp-2} Redundancy
        \quad
        \ref{leg:sim-results-small-ttc-redundantretry-02-amp-2-10} Redundant-Retry(10)
    \end{center}
    \vspace{-0.5em}
    \begin{subfigure}{\linewidth}
        \centering
        \begin{tikzpicture}
            \begin{axis}[
                    myresultsplot01,
                    ylabel={Avg. Success Throughput\\per Node [\si{\xrp\per\second\per\nde}]},
                    xlabel={Max. Number of Additional TXs ($u$)},
                    symbolic x coords={0,2,4,6,8,10,15,20,25,30,40,50,75,100,125,150},
                    height=0.4\linewidth,   %
                    label style={align=center},
                    xlabel={},
                ]

                \def\PLOTDATAFILE{figures/02_nodes100_txs500_paths25_edges4.605170to6.907755_paths25.data}

                \foreach \thisColor/\thisLabel/\thisAlgorithm in {
                    myParula01Blue/Retry/retry-02-amp-2,
                    myParula02Orange/Redundancy/redundancy-02-amp-2,
                    myParula03Yellow/Redundant-Retry(10)/redundantretry-02-amp-2-10} {
                        \edef\temp{
                            \noexpand\addplot+ [
                                error bars/.cd,
                                y explicit,
                                y dir=both,
                            ] table [
                                x=u,
                                y expr=\noexpand\thisrow{throughput_success-mean},
                                y error expr=\noexpand\thisrow{throughput_success-std},
                                discard if not={algo}{\thisAlgorithm},
                            ] {\PLOTDATAFILE};
                            \noexpand\addlegendentry{\thisLabel};
                        }
                        \temp
                }

                \legend{};

            \end{axis}
        \end{tikzpicture}
        \caption{
        `Redundancy' and `Redundant-Retry(10)'
        show a 2x increase in throughput
        compared to `Retry'.
        }
        \label{fig:sim-results-throughput}
        \medskip
    \end{subfigure}
    \begin{subfigure}{\linewidth}
        \centering
        \begin{tikzpicture}
            \begin{axis}[
                    myresultsplot01,
                    ylabel={Avg. Time-To-Completion\\for Successful TFs [\si{\second}]},
                    xlabel={Max. Number of Additional TXs ($u$)},
                    symbolic x coords={0,2,4,6,8,10,15,20,25,30,40,50,75,100,125,150},
                    height=0.4\linewidth,   %
                    label style={align=center},
                    xlabel={},
                ]

                \def\PLOTDATAFILE{figures/02_nodes100_txs500_paths25_edges4.605170to6.907755_paths25.data}

                \foreach \thisColor/\thisLabel/\thisAlgorithm in {
                    myParula01Blue/Retry/retry-02-amp-2,
                    myParula02Orange/Redundancy/redundancy-02-amp-2,
                    myParula03Yellow/Redundant-Retry(10)/redundantretry-02-amp-2-10} {
                        \edef\temp{
                            \noexpand\addplot+ [
                                error bars/.cd,
                                y explicit,
                                y dir=both,
                            ] table [
                                x=u,
                                y expr=\noexpand\thisrow{ttc_for_successful_tx-mean},
                                y error expr=\noexpand\thisrow{ttc_for_successful_tx-std},
                                discard if not={algo}{\thisAlgorithm},
                            ] {\PLOTDATAFILE};
                            \noexpand\addlegendentry{\thisLabel};
                        }
                        \temp
                }

                \legend{};

            \end{axis}
        \end{tikzpicture}
        \caption{
        `Redundancy'
        shows a \SI{40}{\percent} reduction in TTC
        over `Retry'.
        `Redundant-Retry(10)' interpolates
        in TTC
        between
        `Retry' and `Redundancy'.
        }
        \label{fig:sim-results-ttc}
        \medskip
    \end{subfigure}
    \begin{subfigure}{\linewidth}
        \centering
        \begin{tikzpicture}
            \begin{axis}[
                    myresultsplot01,
                    ylabel={Avg. Success Volume\\for Successful TFs [\si{\xrp}]},
                    xlabel={Max. Number of Additional TXs ($u$)},
                    symbolic x coords={0,2,4,6,8,10,15,20,25,30,40,50,75,100,125,150},
                    height=0.4\linewidth,   %
                    label style={align=center},
                    xlabel={},
                ]

                \def\PLOTDATAFILE{figures/02_nodes100_txs500_paths25_edges4.605170to6.907755_paths25.data}

                \foreach \thisColor/\thisLabel/\thisAlgorithm in {
                    myParula01Blue/Retry/retry-02-amp-2,
                    myParula02Orange/Redundancy/redundancy-02-amp-2,
                    myParula03Yellow/Redundant-Retry(10)/redundantretry-02-amp-2-10} {
                        \edef\temp{
                            \noexpand\addplot+ [
                                error bars/.cd,
                                y explicit,
                                y dir=both,
                            ] table [
                                x=u,
                                y expr=\noexpand\thisrow{volume_for_successful_tx-mean},
                                y error expr=\noexpand\thisrow{volume_for_successful_tx-std},
                                discard if not={algo}{\thisAlgorithm},
                            ] {\PLOTDATAFILE};
                            \noexpand\addlegendentry{\thisLabel};
                        }
                        \temp
                }

                \legend{};

            \end{axis}
        \end{tikzpicture}
        \caption{
        `Redundant-Retry(10)' satisfies
        \SI{30}{\percent} larger TFs than `Redundancy'
        which in turn satisfies
        \SI{40}{\percent} larger TFs than the `Retry' baseline.
        }
        \label{fig:sim-results-size}
    \end{subfigure}
    \caption{
        Performance
        as a function of maximum number of additional TXs ($u$)
    }
    \label{fig:sim-results}
\end{figure}

}

\end{document}